\documentclass[lettersize,journal]{IEEEtran}
\pdfoutput=1
\usepackage{amsmath,amsfonts}
\usepackage{algorithmic}
\usepackage{algorithm}
\usepackage{array}
\usepackage[caption=false,font=normalsize,labelfont=sf,textfont=sf]{subfig}
\usepackage{textcomp}
\usepackage{stfloats}
\usepackage{url}
\usepackage{verbatim}
\usepackage{amsthm}
\usepackage{graphicx}
\usepackage{cite}
\usepackage{xcolor}

\theoremstyle{plain}
\newtheorem{proposition}{Proposition}
\begin{document}

\title{Geometry-Aware Resource Allocation for Network-Level ISAC Systems}

\author{Xiao-Yang~Wang,~\IEEEmembership{Member,~IEEE},~Luting~Kong,~Lei~Cao,~Yang~Liu,~Jingheng~Zheng,~\IEEEmembership{Member,~IEEE},~Weiwen~Weng,~Weiyan~Chen, Wenzhi Li,~Kaitao Meng,~\IEEEmembership{Member,~IEEE}, ~Christos~Masouros,~\IEEEmembership{Fellow,~IEEE}
	
    \thanks{X. -Y. Wang, L. Kong, L. Cao, Y. Liu,  J. Zheng, W. Weng and W. Chen are with the Department of Wireless and Device Technology Research, China Mobile Research Institute, Beijing 100053, China (e-mail: wangxiaoyangyjy@chinamobile.com).}

    \thanks{W. Li is with Department of Science and Technology Innovation, China Mobile Communications Group Co., Ltd., Beijing 100033, China (e-mail: liwenzhijs@chinamobile.com).}

    \thanks{K. Meng is with the Department of Electrical and Electronic Engineering, University of Manchester, Manchester M13 9PL, UK (e-mail: kaitao.meng@manchester.ac.uk).}
    
	\thanks{C. Masouros is with the Department of Electronic and Electrical Engineering, University College London, London WC1E 7JE, UK (e-mail: c.masouros@ucl.ac.uk).}
}

\markboth{Submitted to IEEE Transactions on Wireless Communications, June~2026}%
{Shell \MakeLowercase{\textit{et al.}}: A Sample Article Using IEEEtran.cls for IEEE Journals}


\maketitle

\begin{abstract}
Network-level integrated sensing and communication (ISAC) is recognized as a transformative technology for next-generation mobile radio systems. By enabling collaboration among multiple transceivers, network-level ISAC can significantly enhance both communication and sensing performance through spatial diversity. However, existing resource allocation strategies typically overlook the impact of spatial geometry, where identical time-frequency resources contribute differently to sensing accuracy depending on the transceiver's location. This leaves the fundamental coupling between spatial topology and resource efficacy unclear, rendering optimal resource allocation a critical challenge for unlocking the full potential of network-level ISAC.
To address this challenge, this paper investigates the optimal distribution of time-frequency resources across spatially distributed transceivers through a theoretically grounded two-stage framework. First, we analytically derive the optimal time and frequency aperture distributions for sensing, defined as the variances of the allocated symbol and subcarrier indices, respectively, under both two-transmitter and multi-transmitter scenarios. By exploiting the mathematical isomorphism between delay and Doppler estimation, we prove that the optimal resource allocation strategy follows the gradient direction of the Cramér-Rao Lower Bound (CRLB) with respect to the apertures. Second, to bridge the gap between theoretical aperture values and practical OFDMA constraints, such as the minimized communication rate of each user equipment (UE), we formulate the resource allocation as a combinatorial integer partitioning problem. To tackle the NP-hard nature of the formulated problem, a low-complexity Variance-Guided Partitioning Algorithm (VGPA) is proposed to jointly optimize the subcarrier and symbol patterns for communication and sensing. Numerical results validate that the proposed scheme approaches the theoretical performance lower bound and can significantly improve sensing performance while satisfying the communication rate for each UE.
\end{abstract}


\section{Introduction}
\IEEEPARstart{I}{ntegrated} sensing and communication (ISAC) is deemed a key enabling technology for next-generation mobile systems \cite{zhang2021overview,wei2022toward,10292797}. By sharing spectrum, algorithms, and hardware, ISAC systems can reduce the cost of mobile communication infrastructure \cite{8999605}, while achieving highly efficient sensing capabilities that support many promising applications, such as indoor intrusion detection \cite{11205178}, environmental reconstruction \cite{9534682,10103813}, and low-altitude UAV regulation \cite{11098638}. The ultimate goal of deploying ISAC in cellular networks is to enable distributed sensing of an unprecedented scale \cite{10726912}. However, most existing ISAC studies mainly focus on link-level systems design \cite{11184506}. Since practical mobile networks are inherently organized as ubiquitous distributed nodes, extending ISAC to a cooperative network level is an inevitable necessity to fully utilize this natural infrastructure. The failure to fully exploit this paradigm limits the development of large-scale perception applications \cite{meng2024,10438975}.

To enable distributed sensing over ISAC networks, network-level ISAC, also referred to as cooperative or distributed ISAC, has been proposed. In this paradigm, multiple distributed ISAC nodes collaboratively perform signal transmission, reception, and  processing to enhance  performance \cite{11184506}. The collaboration can span across multiple functions, such as synchronization \cite{10091198,wxy}, coordinated beamforming \cite{meng2024}, interference mitigation \cite{10735119,11075613}, and cooperative environment sensing \cite{10557620}. 
Regardless of the collaboration form, network-level ISAC inherently involves resource allocation and coordination across distributed nodes \cite{10614082}. For example, coordinated beamforming corresponds to spatial-domain collaboration \cite{meng2024network}, while interference mitigation requires joint optimization of time-, frequency-, and spatial-domain resources in current OFDMA systems \cite{10663785}. Moreover, unlike purely communication-centric networks where resources are allocated based on Channel State Information (CSI), ISAC networks must also consider the sensing geometry. The contribution of a specific node to the localization accuracy depends heavily on its spatial relationship with the target and the receiver. Consequently, resource distribution becomes a fundamental issue in network-level ISAC, as its design directly affects the sensing capability \cite{wang2024fundamentaltradeoffstimefrequencyresource}. 

While prior works have optimized pilot patterns and power allocation for single-node ISAC systems \cite{10288116, 8628347, 9062788}, the joint optimization of spatial topology and time-frequency domain resources in distributed networks remains unexplored. For example, the authors in \cite{10288116} optimize the frequency-domain distribution of the demodulation reference signal (DMRS) for narrowband systems, demonstrating that such optimization improves range sensing performance. However, this study considers only resource-block-group-level Type-0 DMRS distribution, one of three DMRS distribution schemes in the 5G standard \cite{3gpp.38.214}, and focuses solely on improving range estimation CRLB. Similarly, the authors of \cite{8628347} design a new frequency-domain pilot pattern, termed the stepped pattern, which extends the maximum perceptible range while maintaining a constant range estimation Cramér-Rao Lower Bound (CRLB). This design is further validated by hardware implementation in \cite{9062788}. Nevertheless, the pilot distribution in this study targets only the maximum perceptible range and does not address improvements in range CRLB. 

To further investigate the fundamental tradeoffs between the time- and frequency-domain resource distribution and sensing performance, the authors of \cite{wang2024fundamentaltradeoffstimefrequencyresource} jointly optimize time- and frequency-domain resource allocation for single-user equipment (UE) OFDMA and multi-UE OFDMA systems (namely network-level ISAC systems), respectively. They identify the best and worst distribution schemes for both scenarios. However, these studies do not account for the spatial distribution of transmitters, which is known to significantly influence positioning performance in wireless sensor networks \cite{5466526,236507}. 

Despite these advances, no existing work jointly considers time-domain, frequency-domain, and spatial distribution of transceivers, ignoring that transceivers in favorable geometric positions should theoretically be allocated ``wider" effective bandwidths or time durations (apertures) to maximize sensing, while satisfying constraints on communication rate.
To fully exploit the sensing potential of OFDMA-based network-level ISAC, this paper investigates the fundamental coupling between spatial transceiver distribution and time-frequency resource allocation. We propose a systematic optimization framework that bridges the gap between sensing geometry and resource distribution with communication constraints. Specifically, we first identify the ``frequency/time aperture'' as the critical metric determining sensing accuracy. Based on this metric, we derive theoretical optimal distributions and subsequently design practical allocation algorithms. The main contributions of this work are summarized as follows:

\begin{itemize}
    \item {Theoretical framework via geometric water-filling.} We establish a closed-form solution for the optimal aperture allocation in two-transmitter OFDMA systems. For general multi-transmitter networks, utilizing asymptotic analysis, we reveal a fundamental \textit{``geometric water-filling"} principle: akin to power allocation in communications-only cooperation, frequency/time resources should be preferentially ``poured" into transceivers with higher ``geometric gains" (i.e., CRLB sensitivity), while nodes with poor sensing geometry are assigned minimal resources.

    \item {Unified time-frequency analysis.} We explicitly define ``frequency aperture" and ``time aperture" and exploit the mathematical isomorphism between delay and Doppler estimation. This allows us to unify the resource optimization for both ranging and velocity sensing under a single mathematical framework, significantly generalizing the applicability of our proposed method.

    \item {Efficient aperture allocation algorithm design.} To bridge the gap between the theoretical optimal aperture (a continuous value) and practical constraints, such as the required communication rate, we formulate the subcarrier/symbol allocation as a variance-constrained integer partitioning problem. Since this problem is NP-hard, we develop a low-complexity Variance-Guided Partitioning Algorithm (VGPA). Numerical results demonstrate that the VGPA effectively approaches the theoretical performance lower bound in two-UE scenarios and significantly improve sensing performance with required communication rates.
\end{itemize}

The rest of this paper is organized as follows. In Section II, we present the system model for OFDMA network-level ISAC systems. In Section III, we optimize the sensing performance in terms of the frequency aperture for both two-UE scenarios and multi-UE scenarios. Section IV formulates the subcarrier allocation problem and presents the corresponding algorithm to allocate subcarriers according to the prescribed frequency apertures. Numerical results are presented in Section V to corroborate our analysis. Finally, conclusions are drawn in Section VI.

\textit{Notations}: ${\bf A}^{\textrm{T}},{\bf A}^{*}, {\bf A}^{\textrm{H}}$, and ${\bf A}^{-1}$ represent the transpose, conjugate, conjugate transpose, and inverse of the matrix ${\bf A}$, respectively. ${\bf a}[m]$ and ${\bf A}[m,n]$ are the $m$th element of the vector ${\bf a}$ and the $(m,n)$th element of the {matrix} ${\bf A}$, respectively; ${a}(t)$,  ${\bf a}(t)$, and ${\bf A}(t)$ are the scalar function, the vector function, and the matrix function with respect to $t$, respectively; ${a}({\bf t})$ is the scalar function with regard to vector ${\bf t}$; $\mathbb{E}[\cdot]$ and $\operatorname{Tr}(\cdot)$ denote the expectation and trace operator, respectively; {while $j$  denotes the imaginary unit; Moreover, $|\cdot|$ represents the number of elements in a set, the modulus of a vector, or the absolute value of a scalar. $\mathbb{Z}$, $\mathbb{Z}^{+}$ and $\mathbb{R}$ are defined as the set of integer, the non-negative integer and real number, respectively; $\cap$ and $\cup$ are the set intersection and union operators, respectively; Finally, $\nabla$ and $\partial $ represent the gradient and the Partial derivative of a function, respectively.

\section{System Model}

\begin{figure}[tbp]
	\centering
	\includegraphics[width=7cm]{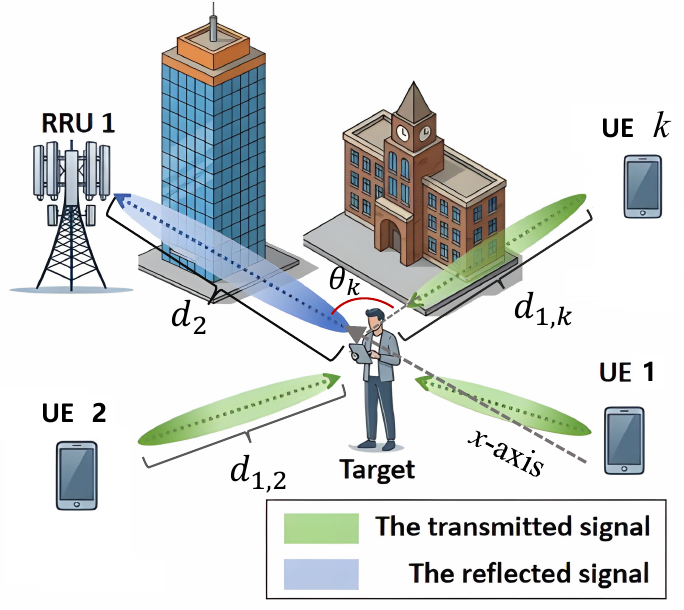}
	\caption{System model.}
	\label{figure1}
\end{figure}

Consider an uplink MIMO-OFDMA ISAC system, as depicted in Fig. \ref{figure1}. The $M_{\textrm{R}}$-antenna RRU receives signals transmitted by $K$ UEs, each equipped with $M_{\textrm{T}}$ antennas. All $K$ UEs cooperatively illuminate a common target at coordinates ${\bf p}=[p_x,p_y]^\mathrm{T}$ and the RRU estimate the target parameters in a multi-static manner.
Our primary objective is to optimize the sensing accuracy of this target by strategically allocating time-frequency resources among the distributed UEs based on their spatial geometry. In order to fully demonstrate the role of geometric topology, we focus here primarily on passive sensing in network-level ISAC.
Specifically, let  $d_{1,k}$ denote the distance between the $k$th UE and the target, and $d_2$ denote the distance between the target and the RRU. The propagation delay from the $k$th UE to the target and then to the RRU is denoted as $\tau_k$. The RRU can estimate the target position using the received uplink signals by firstly demodulating data and then processing sensing information, as described in \cite{rahman2019framework,zhang2021enabling}, or by joint data demodulation and target sensing, as depicted in \cite{valiulahi,10918629}. 
Let $f_\textrm{c}$, $N$, $\Delta f$, and $T_\textrm{sym}$ denote the carrier frequency, total number of subcarriers, subcarrier spacing, and OFDM symbol duration, respectively. Each UE is assigned a subset of frequency subcarrier indices given by the set $\mathcal{S}_k$ with cardinality $N_k$ and a subset of symbol time-slot indices given by $\boldsymbol{\mathcal{T}}_k$. These sets $\mathcal{S}_k$ and $\boldsymbol{\mathcal{T}}_k$ are the key optimization variables in our study, representing the frequency and time resources allocated to the $k$-th UE, respectively. Moreover, the communication data symbols modulated on these subcarriers in the $g$th OFDM symbol are denoted as ${\bf s}_{g,k}\in\mathbb{C}^{1\times N_k}$. 

Therefore, the equivalent baseband signal in $g$th OFDM symbol transmitted by the $k$th UE can be expressed as
\begin{equation}\label{x_k(t)}
{\bf x}_{g,k}(t)=e^{j2\pi f_ct}{{\bf w}_k}\sum\nolimits_{{n \in \mathcal{S}_k}}{\bf s}_{g,k}[n]e^{j2\pi {n}\Delta ft}.
\end{equation} where ${\bf w}_k\in\mathbb{C}^{M_{\textrm{T}}\times 1}$ denotes the beamforming vector of the $k$th UE. Without loss of generality, we assume that ${\bf w}_k$ is designed to align with the target direction (e.g., via maximum ratio transmission) to maximize the effective sensing gain. This allows us to focus on the optimization of time-frequency resources while communication proceeds simultaneously.

Accordingly, the received \textit{time-domain} equivalent baseband signal corresponding to the $g$th symbol can be expressed as
\begin{equation}
\begin{aligned}
\mathbf{y}_{g}(t)&=\sum_{k=1}^K\sigma{d}_{1,k}^{-\frac\beta2}{d}_2^{-\frac\beta2}\mathbf{a}_\textrm{R}\left(\theta_0\right)\mathbf{a}_\textrm{T}^{\rm H}\left(\theta_{k}\right){\bf x}_{g,k}(t-\tau_k)+\mathbf{w}(t)\\
&=\sum_{k=1}^K\sum_{{n \in \mathcal{S}_k}}e^{j2\pi f_c(t-\tau_k)}\sigma{d}_{1,k}^{-\frac\beta2}{d}_2^{-\frac\beta2}\mathbf{a}_\textrm{R}\left(\theta_0\right)\mathbf{a}_\textrm{T}^H\!\!\left(\theta_{k}\right)\!{{\bf w}_k}\\
&\ \ \ \cdot{\bf s}_{g,k}[n]e^{j2\pi {n}\Delta f(t-\tau_k)}\!+\!\mathbf{w}(t),
\end{aligned}
\end{equation}
where $\theta_{k}$ denotes the angle formed by the link from the $k$th UE to the target and the reflected link from the target to the RRU, which overlaps with the defined x-axis, and $\theta_0=0$ is the angle formed by the link from the target to the RRU and the x-axis, while ${\bf a}_\textrm{R}(\cdot)\in \mathbb{C}^{M_\textrm{R}\times 1}$ and ${\bf a}_\textrm{T}^{\rm H}(\cdot)\in \mathbb{C}^{1\times M_\textrm{T}}$ are the receiving and transmitting steering vectors, respectively. $r_\textrm{c}$ and $\beta$ represent the radar cross section
(RCS) of the target and the path-loss exponent of the environment, respectively. $\mathbf{w}(t)$ denotes the additive white Gaussian noise (AWGN) satisfying $\mathbb{E}[\mathbf{w}^2(t)]=z_n^2$.

\begin{figure}[tbp]
	\centering
	\includegraphics[width=8.4cm]{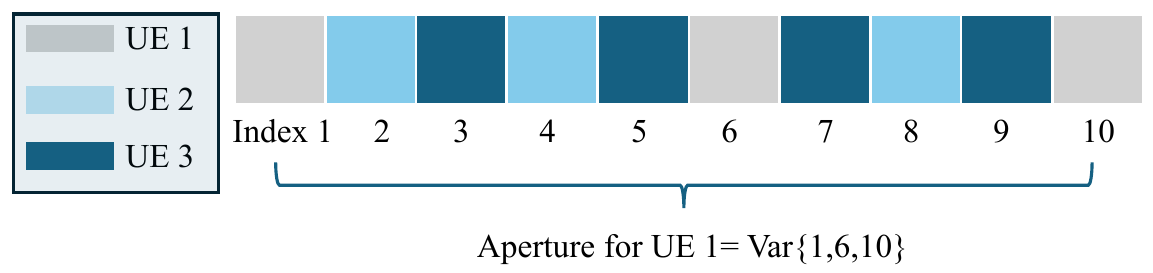}
	\caption{Frequency aperture computation for UE 1. The colorful boxes represent the subcarrier.}
	\label{figure2new}
\end{figure}

In the ISAC receiver, the transmitted data symbols are first recovered via standard demodulation and decoding processes. These recovered symbols then serve as known reference signals for the sensing module for simplicity \cite{rahman2019framework,10640151}. Then, the target location parameter $(p_x,p_y)$ can be estimated using the super-resolution algorithm or the maximum likelihood estimation method \cite{5466526}. To characterize the fundamental sensing performance limit, the Fisher information matrix (FIM) for estimating ${\bf p}$ is given by \cite{meng2024,9364752} 
\begin{equation}\label{equ3}
\mathbf{F}_N=\rho \sum_{k=1}^K {d}_{1,k}^{-\beta}{d}_2^{-\beta} \sigma_{k}^2 \left[\begin{array}{cc}
a_{k}^2 & a_{k} b_{k} \\
a_{k} b_{k} & b_{k}^2
\end{array}\right]
\end{equation} where $\sigma_{k}^2$ denotes the variance of the elements in $\mathcal{S}_k$, and $\rho$ is given by
\begin{equation}
\rho=\frac{G_t G_r \Delta f^2 r_\textrm{c}}{8 \pi f_c^2 z_n^2},
\end{equation} where $G_t$ and $G_r$ denote the transmit and receive beamforming gains, respectively. Moreover, $a_{k}$ and $b_{k}$ are expressed as \cite{meng2024}
%
\begin{equation}
\begin{aligned}
a_{k} & =\cos \theta_{k}+\cos \theta_0, \\
b_{k} & =\sin \theta_{k}+\sin \theta_0.
\end{aligned}
\end{equation}

For simplicity, by assuming $\theta_0=0$, the CRLB can be given as {$\textrm{CRLB}=\operatorname{Tr}(\mathbf{F}_N^{-1})$}, as shown in (\ref{equ7}), where {$v_{k_1}={d}_{1,k}^{-\beta}{d}_2^{-\beta} \sigma_{k_1}^2$}.
\begin{figure*}
\begin{equation}\label{equ7}
\begin{aligned}
\textrm{CRLB}=\frac{\sum_{k=1}^K v_k(1+\cos(\theta_k))^2+\sum_{k=1}^K  v_k\sin^2(\theta_k)}{\rho\left(\sum_{k=1}^K v_k(1+\cos(\theta_k))^2\right)\left(\sum_{k=1}^Kv_k \sin^2(\theta_k)\right)-\left(\sum_{k=1}^K v_k(1+\cos(\theta_k)) \sin(\theta_k)\right)^2}.
\end{aligned}
\end{equation}
\end{figure*}

Beyond location, velocity estimation accuracy depends on time-domain spread \cite{wang2024fundamentaltradeoffstimefrequencyresource}. Recognizing that the FIM components scale with the variance of resource indices \cite{wang2024fundamentaltradeoffstimefrequencyresource}, we unifiedly define the \textit{frequency aperture} $\mathcal{A}^\mathrm{freq}_k$ and \textit{time aperture} $\mathcal{A}^\mathrm{time}_k$ for the $k$-th UE as $\mathcal{A}^\mathrm{freq}_k = \text{Var}(\mathcal{S}_k)$ and $\mathcal{A}^\mathrm{time}_k = \text{Var}(\boldsymbol {\mathcal{T}}_k)$, respectively. For ease of understanding, we show the definition of frequency aperture in Fig. \ref{figure2new}.

{\bf Remark}:
Leveraging the \textit{time-frequency duality} of OFDM \cite{liu2020super}, delay and Doppler estimation are \textit{mathematically isomorphic}. Thus, velocity optimization is structurally identical to ranging. Consequently, our subsequent derivation focuses on the frequency domain, yet the resulting analysis and algorithms apply universally to the time domain. 

\section{Subcarrier Distribution Optimization}

In this section, we investigate the optimal allocation of frequency aperture resources to minimize the sensing CRLB. We begin by introducing the overall theoretical solution framework in Subsection~A. Next, for frequency aperture optimization, we investigate the special case of two transmitters in Subsection~B to reveal fundamental geometric properties and establish geometric intuition. Building on these insights, we generalize the analysis to multi-transmitter systems using asymptotic analysis in Subsection~C. Finally, the practical subcarrier allocation corresponding to the optimal frequency aperture with constraints on communication rate for each UE via discrete integer partitioning is addressed in Section IV.

\subsection{Problem Analysis and Solution Framework}
For simplicity, this subsection outlines the core solution framework. Specifically, the optimization problem formulated in this Subsection (\ref{equ10-1}) presents two fundamental challenges:
\begin{itemize}
    \item {Combinatorial complexity:} the optimization variables are discrete sets of subcarrier indices $\{\mathcal{S}_k\}$, making the problem a non-convex integer programming problem which is generally NP-hard.
    \item {Matrix coupling:} the objective function $\operatorname{Tr}(\mathbf{F}_N^{-1})$ involves the inversion of a sum of matrices. For the multi-transmitter case ($K>2$), the contributions of individual UEs are coupled inside the inverse operation, preventing the derivation of a straightforward closed-form solution.
\end{itemize}

To overcome these hurdles, we propose a {two-stage theoretical framework}, decoupling the determination of optimal aperture sizes from the discrete subcarrier assignment:

\subsubsection*{Stage 1: Continuous Aperture Optimization (Section \uppercase\expandafter{\romannumeral3})}
We first relax the discrete integer constraints and treat the aperture variance $\sigma_k^2$ as a continuous variable. To reveal the relationship between spatial geometry and optimal variance, we adopt a {progressive analytical strategy}:
\begin{itemize}
    \item {Baseline insight (Subsection III-B):} We start with the two-transmitter case ($K=2$), where the matrix inversion can be explicitly expanded. This yields a closed-form solution that provides intuitive ``water-filling'' insights.
    \item {Asymptotic generalization (Subsection III-C):} For the general multi-transmitter case ($K>2$), we utilize asymptotic analysis based on the Law of Large Numbers. By deriving the gradient of the expected CRLB, we establish the \textit{geometric water-filling (GWF)} principle for large-scale networks.
\end{itemize}

\subsubsection*{Stage 2: Discrete Integer Partitioning (Section IV)}
Guided by the theoretical optimal variances derived in Stage 1, we then address the integer constraints by designing a low-complexity algorithm (VGPA) to map continuous variances to discrete subcarrier sets.

\subsection{Two-Transmitter Scenarios}

For the two-transmitter case ($K=2$), the FIM $\mathbf{F}_N$ in (\ref{equ3}) can be expressed as
\begin{equation}\label{equ3-1}
\mathbf{F}_N = \rho \left( v_1 \begin{bmatrix} a_1^2 & a_1 b_1 \\ a_1 b_1 & b_1^2 \end{bmatrix} + v_2 \begin{bmatrix} a_2^2 & a_2 b_2 \\ a_2 b_2 & b_2^2 \end{bmatrix} \right).
\end{equation}

Then, the CRLB, defined as $\operatorname{Tr}(\mathbf{F}_N^{-1})$, can be calculated as
\begin{equation}
\operatorname{Tr}(\mathbf{F}_N^{-1}) = \frac{\operatorname{Tr}(\mathbf{F}_N)}{\det(\mathbf{F}_N)}.
\end{equation}

By defining $L_k = a_k^2 + b_k^2 = 2(1+\cos\theta_k)$, the trace of $\mathbf{F}_N$ is given by $\rho(v_1 L_1 + v_2 L_2)$. The determinant is calculated as
\begin{equation}
\det(\mathbf{F}_N) = \rho^2 v_1 v_2 (a_1 b_2 - a_2 b_1)^2.
\end{equation}
Then, we define $\Delta_{1,2}$ as
\begin{equation}
\Delta_{1,2} = (a_1 b_2 - a_2 b_1)^2 = [\sin(\theta_2) - \sin(\theta_1) + \sin(\theta_2-\theta_1)]^2.
\end{equation}

Consequently, the CRLB can be simplified as
\begin{equation}\label{equ8-3}
\begin{aligned}
\operatorname{Tr}(\mathbf{F}_N^{-1}) &= \frac{\rho(v_1 L_1 + v_2 L_2)}{\rho^2 v_1 v_2 \Delta_{1,2}} \\
&= \frac{1}{\rho \Delta_{1,2}} \left( \frac{L_1}{v_2} + \frac{L_2}{v_1} \right).
\end{aligned}
\end{equation}
(\ref{equ8-3}) reveals a fundamental cooperative coupling effect: the contribution of the $1$st UE's geometric quality $L_1$ is scaled by the resource intensity of the $2$nd UE $v_2$, and vice versa. This implies that to fully exploit the superior angle of one node, the other node must provide sufficient signal quality to resolve the ambiguity.

To optimize the two-dimensional (2D) location accuracy, $\textrm{Tr}(\mathbf{F}_N^{-1})$ should be minimized. Moreover, according to \cite{wang2024fundamentaltradeoffstimefrequencyresource}, ${\sigma}_{1}^2$ and ${\sigma}_{2}^2$ are required to satisfy 
\begin{equation}
N_1\sigma_{1}^2+N_2\sigma_{2}^2<\sum_{n=1}^{N}(n-\frac{1}{N}\sum_{m=1}^{N}m)^2.
\end{equation} Consequently, the optimization problem is then formulated as
\begin{equation}\label{equ10-1}
\begin{aligned}
\mathop{\min}_{{{\mathcal{S}_1}, {\mathcal{S}_2}}}\ \
& \frac{L_1}{d_{1,2}^{-\beta}d_2^{-\beta}\sigma_2^2} + \frac{L_2}{d_{1,1}^{-\beta}d_2^{-\beta}\sigma_1^2}\\
\textrm{s.t.}\ \ &N_1{\sigma_{1}^2}+N_2{\sigma_{2}^2} \le E_{\text{total}},\\
&{\sigma_{k}^2}=\frac{1}{N_k}\sum_{{n\in\mathcal{S}_k}}({n}-\frac{1}{N_k}\sum_{{m\in\mathcal{S}_k}}{m})^2, {k=1,2}\\
&{\mathcal{S}_k}\subset \mathcal{N}, |{\mathcal{S}_k}|=N_k, {k=1,2}\\
&N_1, N_2 \in \mathbb{Z^{+}}.
\end{aligned}
\end{equation}
where $E_{\text{total}} = \sum_{n=1}^{N}(n-\frac{1}{N}\sum_{m=1}^{N}m)^2$ is the total variance budget. As shown in (\ref{equ10-1}), this problem falls into the category of integer programming \cite{wang2024fundamentaltradeoffstimefrequencyresource}, which is inherently difficult to solve optimally. To tackle this issue, we first optimize the objective with respect to $\sigma_{1}^2$ and $\sigma_{1}^2$, and then design the corresponding subcarrier distributions based on the obtained optimal values. Specifically, the first part is solved in this subsection, while the design of the corresponding subcarrier distributions is presented in Section \uppercase\expandafter{\romannumeral4}, which also optimizes the subcarrier distribution scheme for multi-transmitter scenarios.

Firstly, the problem with regard to $\sigma_{1}^2$ and $\sigma_{2}^2$ can be formulated as
\begin{equation}\label{equ10-2}
\begin{aligned}
\mathop{\min}_{{\sigma_{1}^2}, {\sigma_{2}^2}}\ & {\operatorname{Tr}(\mathbf{F}_N^{-1})}\\
\textrm{s.t.}\ \ &N_1{\sigma_{1}^2}+N_2{\sigma_{2}^2}\le\sum_{n=1}^{N}(n-\frac{1}{N}\sum_{m=1}^{N}m)^2,\\
&{\sigma_{1}^2}\in\mathbb{R}, {\sigma_{2}^2}\in\mathbb{R},\\
&N_1, N_2 \in \mathbb{Z^{+}}.
\end{aligned}
\end{equation}
Obviously, the objective is convex with regard to $\sigma_{1}^2$ and $\sigma_{2}^2$ according to \cite{boyd2004convex}. Therefore, the optimal solution can be obtained using the Lagrange multiplier method \cite{bertsekas2014constrained}. Specifically, we have the following proposition.

\begin{proposition}\label{prop1}
Given the resource constraint $E_{\text{total}}$, the optimal frequency apertures $\sigma_1^2$ and $\sigma_2^2$ that minimize the CRLB are given by
\begin{align}
\sigma_1^2 &= \frac{E_{\text{total}}}{N_1 + \sqrt{\frac{N_1 N_2 \Gamma_2}{\Gamma_1}}}, \label{opt_sig1}\\
\sigma_2^2 &= \frac{E_{\text{total}}}{N_2 + \sqrt{\frac{N_1 N_2 \Gamma_1}{\Gamma_2}}}, \label{opt_sig2}
\end{align}
where $\Gamma_1 = \frac{L_2}{d_{1,1}^{-\beta}d_2^{-\beta}}$ and $\Gamma_2 = \frac{L_1}{d_{1,2}^{-\beta}d_2^{-\beta}}$ represent the cross-coupled geometric-channel gains.
\end{proposition}

\textit{Proof:} See Appendix A.

\textbf{Remark:} Proposition \ref{prop1} demonstrates a \textit{geometric water-filling} strategy. Unlike traditional communication-only water-filling where resource ratios are determined by decoupled communication-channel qualities, here $\sigma_1^2$ scales with $\Gamma_1$, which explicitly incorporates $L_2$ (the angular quality of the partner UE 2). This confirms that in network-level ISAC, resources are allocated to a node not merely to enhance its local signal strength, but to cross-leverage the geometric potential of the multi-static transmitter.

\subsection{Generalization to Multi-Transmitter Networks}

For general multi-transmitter networks ($K>2$), deriving a closed-form solution is mathematically intractable due to the non-linear matrix inversion in the objective $\operatorname{Tr}((\sum_{k=1}^{K} \mathbf{F}_k)^{-1})$, which becomes significantly more challenging. Unlike the two-transmitter case, this inversion creates a complex coupling among all user resources, preventing variable decoupling. Consequently, the marginal gain of any single user depends on the entire network configuration. To overcome this hurdle, we adopt a gradient-based asymptotic analysis approach. Our solution logic proceeds as follows: first, we derive the asymptotic expression of the CRLB's partial derivative with respect to the frequency aperture,  for large-scale networks using the law of large numbers; second, based on this explicit gradient, we mathematically prove the monotonicity of the CRLB and identify the direction of steepest descent. Then, aligning with this gradient direction, we establish the Geometric Water-Filling resource allocation strategy.

\subsubsection{Asymptotic Gradient of the CRLB}
For formula simplicity, we first define $\mathbf{v} = [v_1, \dots, v_K]^T$, where the $k$-th element is $v_k = {d_{1,k}^{-\beta}d_2^{-\beta}} \cdot \sigma_k^2$. Unlike traditional communication water-filling strategies that allocate resources based solely on channel gain (path loss), sensing performance relies on the coupling of signal strength and geometric spread. Here, $v_k$ integrates the physical channel condition with the tunable geometric resource (aperture), serving as the core variable to derive a unified gradient that balances both channel depth and angular diversity. Then, the CRLB in (\ref{equ7}) can be expressed as
\begin{equation}
\textrm{CRLB}=\frac{g_1({\bf v})}{\rho g_2 ({\bf v})},
\end{equation}
where $g_1({\bf v})$ and $g_2 ({\bf v})$ are defined as
\begin{equation}\label{equ8}
g_1({\bf v})=\sum_{k=1}^K 2\left(1+\cos \theta_k\right) v_k,
\end{equation}
\begin{equation}\label{equ9}
\begin{aligned}
& g_2 ({\bf v})=\sum_{k=1}^K v_k a_k^2 \sum_{k=1}^K v_k b_k^2-\left(\sum_{k=1}^K v_k a_k b_k\right)^2 \\
&= \sum_{k=1}^K \sum_{k^{'}=1}^K\left[a_{k}^2 b_{k^{'}}^2-a_k a_{k^{'}} b_k b_{k^{'}}\right] v_k v_k^{\prime} \\
& =\sum_{k=1}^K \sum_{k^{'}=1}^K f\left(k, k^{\prime}\right)\left[f\left(k, k^{\prime}\right)-f\left(k^{\prime}, k\right)\right] v_k v_{k^{\prime}}.
&  \\
\end{aligned}
\end{equation}
In (\ref{equ9}), $f\left(k, k^{\prime}\right)$ are defined as
\begin{equation}
\begin{aligned}
f\left(k, k^{\prime}\right)=\left(1+\cos \theta_k\right) \sin \theta_{k^{\prime}}.
\end{aligned}
\end{equation}

Hence, the partial derivatives of CRLB with regard to $v_k$ can be expressed as
\begin{equation}
\begin{aligned}
\frac{\partial~ \textrm{CRLB}}{\partial v_k} =\frac{\frac{\partial g_1({\bf v})}{\partial v_{k_1}} g_2({\bf v})-g_1({\bf v}) \frac{\partial g_2({\bf v})}{\partial v_{k_1}}}{\rho g_2^2({\bf v})}.
\end{aligned}
\end{equation}
Furthermore, according to (\ref{equ8}) and (\ref{equ9}), we have
\begin{equation}\label{equ12}
\frac{\partial g_1({\bf v})}{\partial v_{k_ 1}}=2\left(1+\cos \theta_{k_ 1}\right),
\end{equation}
\begin{equation}\label{equ13}
\frac{\partial g_2({\bf v})}{\partial v_{k_ 1}}=\sum_{k=1}^K v_k \left[f\left(k, k_1\right)-f\left(k_1, k\right)\right]^2.
\end{equation}
Then, according to (\ref{equ8}), (\ref{equ9}), (\ref{equ12}), and (\ref{equ13}), $g_1({\bf v}) \frac{\partial {g}_2({\bf v})}{\partial v_{k_ 1}}$ and $\frac{\partial g_1({\bf v})}{\partial v_{k_1}} g_2({\bf v})$ can be expressed as
\begin{equation}
\begin{aligned}
&g_1({\bf v}) \frac{\partial {g}_2({\bf v})}{\partial v_{k_ 1}}\\
&=\sum _ { k = 1 } ^ { K } 2 ( 1 + \operatorname { c o s } \theta_k) v_{ k } \sum_{k=1}^K v_k\left[f\left(k, k_1\right)-f\left(k_1, k\right)\right]^2\\
&=\sum_{k=1}^K \sum_{k^{'}=1}^K 2\left(1+\cos \theta_k\right)\left[f\left(k^{\prime}, k_1\right)-f\left(k_1, k^{\prime}\right)\right]^2 v_k v_{k^{\prime}},
\end{aligned}
\end{equation}
\begin{equation}
\begin{aligned}
&\frac{\partial g_1({\bf v})}{\partial v_{k_1}} g_2({\bf v})\\
&=\sum_{k=1}^K \sum_{k^{\prime}=1}^K 2\left(1+\operatorname{cos}\theta_{k_1}\right) f\left(k, k^{\prime}\right)\left[f\left(k, k^{\prime}\right)-f\left(k^{\prime}, k\right)\right] v_k v_{k^{'}}.
\end{aligned}
\end{equation}

As a result, $\mathbb{E}[\frac{\partial~ \textrm{CRLB}}{\partial v_k}]$ can be further formulated as
\begin{equation}
{\mathbb{E}[\frac{\partial~ \textrm{CRLB}}{\partial v_k}]=\frac{\mathbb{E}[\frac{\partial g_1({\bf v})}{\partial v_{k_1}} g_2({\bf v})]-\mathbb{E}[g_1({\bf v}) \frac{\partial {g}_2({\bf v})}{\partial v_{k_ 1}}]}{\rho\mathbb{E}[g_2({\bf v})^2]}.}
\end{equation}
According to the Chebyshev's law of large numbers \cite{revesz2014laws}, when the number of UEs is large, $\mathbb{E}[\frac{\partial g_1({\bf v})}{\partial v_{k_1}} g_2({\bf v})]$ and $\mathbb{E}[\frac{\partial g_1({\bf v})}{\partial v_{k_1}} g_2({\bf v})]$ can be approximately formulated as
\begin{equation}\label{equ17}
\begin{aligned}
&{\mathbb{E}[\frac{\partial g_1({\bf v})}{\partial v_{k_1}} g_2({\bf v})]\approx} \sum_{k=1}^K\sum_{k^{\prime}=1}^K2\left(1+\operatorname{cos}\theta_{k_1}\right) \\
&\cdot {\mathbb{E}[f\left(k, k^{\prime}\right)\left[f\left(k, k^{\prime}\right)-f\left(k^{\prime}, k\right)\right]]\mathbb{E}[ v_k v_{k^{'}}],}
\end{aligned}
\end{equation}
\begin{equation}\label{equ18}
\begin{aligned}
&{\mathbb{E}[g_1({\bf v}) \frac{\partial {g}_2({\bf v})}{\partial v_{k_ 1}}]} \\
&=\!\sum_{k=1}^K \sum_{k^{'}=1}^K 2{\mathbb{E}[\left(1+\cos \theta_k\right)\left[f\left(k^{\prime}, k_1\right)\!-\!f\left(k_1, k^{\prime}\right)\right]^2] \mathbb{E}[v_k v_{k^{\prime}}]}\\
&=\!\sum_{k=1}^K \sum_{k^{'}=1}^K2{\mathbb{E}[1+\cos \theta_k]}\{\sin ^2 \theta_{k_1} {\mathbb{E}[(1+\cos \theta_{k^{'}})^2]}\\
&\ \ \ +(1+\cos\theta_{k_1})^2 {\mathbb{E}[\sin ^2 \theta_{k^{'}}]-2 \sin \theta_{k_1} (1+\cos\theta_{k_1})}\\
&\ \ \ \cdot {\mathbb{E}[(1+\cos \theta_{k^{'}}) \sin \theta_{k^{'}}]\}\mathbb{E}[v_k v_{k^{\prime}}]}.
\end{aligned}
\end{equation}

In specific environments, transmitters typically maintain relatively constant angle and distance distributions over extended periods. Assuming that the joint probability density function (PDF) of the UEs' angles and distances is given by $p(x,y)$, (\ref{equ17}) and (\ref{equ18}) can be further simplified into (\ref{equ191}) and (\ref{equ201}), respectively. Consequently, {$\mathbb{E}[\frac{\partial~ \textrm{CRLB}}{\partial v_k}]$} can be expressed as in (\ref{equ212}). According to this expression, when {$\theta_{k_1}=\pi$}, {$\mathbb{E}[\frac{\partial~ \textrm{CRLB}}{\partial v_k}]$} equals zero, indicating that the corresponding transmitter contributes negligibly to target localization. Therefore, the frequency aperture assigned to such a transmitter should be minimized.

{\bf Remark}: 
In this asymptotic regime, the integral term enclosed in curly braces can be viewed as a constant determined solely by $\theta_{k_1}$. This means that the relative contribution of each UE to the CRLB remains fixed once the angular and distance distributions are given. Consequently, the ratio between CRLB components along any two dimensions becomes constant, leading to an invariant gradient direction of {$\mathbb{E}[\textrm{CRLB}]$} across all coordinate points {$\boldsymbol{\sigma}^2$}. Hence, the optimal frequency aperture allocation vector should be aligned with {$\boldsymbol{\sigma}^2$} itself. For better intuition, an illustrative example is provided below.

\begin{figure*}
\begin{equation}\label{equ191}
\begin{aligned}
&{\mathbb{E}[\frac{\partial g_1({\bf v})}{\partial v_{k_1}} g_2({\bf v})]}\approx\sum_{k=1}^K\sum_{k^{\prime}=1}^K2{\sigma_{k}^2\sigma_{k^{\prime}}^2}\left(1\!+\operatorname{cos}\theta_{k_1}\right)\int_{ 0 }^{ 2 \pi }\!\displaystyle\int_{ 0 }^{ \infty}\!p(\theta_k,d_{1,k})\int_{ 0 }^{ 2 \pi}\!\displaystyle\int_{ 0 }^{ \infty }\!p(\theta_{k^{'}},d_{1,k^{'}}) \\
&\cdot[(1+ \cos \theta_k )^2 \sin ^2(\theta_{k^{'}}) - (1+ \cos \theta_k ) (1+ \cos \theta_{k^{'}} ) \sin ( \theta_{k} ) \sin (\theta_{k^{'}} )]{d}_{1,k}^{-\beta}{d}_{1,k^{'}}^{-\beta}{\rm d}\ {d}_{1,k} {\rm d}\ \theta_k {\rm d}\ {d}_{1,k^{'}} {\rm d}\ \theta_{k^{'}},
\end{aligned}
\end{equation}
\end{figure*}
\begin{figure*}
\begin{equation}\label{equ201}
\begin{aligned}
&{\mathbb{E}[g_1({\bf v}) \frac{\partial {g}_2({\bf v})}{\partial v_{k_ 1}}]}\approx\sum_{k=1}^K \sum_{k^{'}=1}^K2{\sigma_{k}^2\sigma_{k^{\prime}}^2}\int_{ 0 }^{ 2 \pi}\!\displaystyle\int_{ 0 }^{ \infty }\!p(\theta_k,d_{1,k})\int_{ 0 }^{ 2 \pi}\!\displaystyle\int_{ 0}^{ \infty}\!p(\theta_{k^{'}},d_{1,k^{'}})(1+\cos \theta_k)\\
&\cdot[\sin ^2 \theta_{k_1} (1+\cos \theta_{k^{'}})^2+(1+\cos\theta_{k_1})^2 \sin ^2 \theta_{k^{'}}-2 \sin \theta_{k_1} (1+\cos\theta_{k_1}) (1+\cos \theta_{k^{'}}) \sin \theta_{k^{'}}]{d}_{1,k}^{-\beta}{d}_{1,k^{'}}^{-\beta}{\rm d}\ {d}_{1,k} {\rm d}\ \theta_k {\rm d}\ {d}_{1,k^\prime} {\rm d}\ \theta_{k^{'}}.
\end{aligned}
\end{equation}
\end{figure*}

\begin{figure*}
\begin{small}
\begin{equation}
\begin{aligned}\label{equ212}
&{\mathbb{E}[\frac{\partial~ \textrm{CRLB}}{\partial v_{k_1}}]}\approx \Big\{ \int_0^{2\pi} \int_0^\infty p(\theta_k,d_{1,k}) \int_0^{2\pi} \int_0^\infty p(\theta_{k'},d_{1,k'}) \Big[
(1+\cos\theta_{k_1})(1+\cos\theta_k)(\cos\theta_k - \cos\theta_{k_1}) \sin^2\theta_{k'} \!+\!(1+\cos\theta_k)\\
&(1\!+\!\cos\theta_{k'})\sin\theta_k \left(2\sin\theta_{k_1}\!-\!(1\!+\!\cos\theta_{k_1})\sin\theta_{k'} \right)\!-\!(1\!+\!\cos\theta_k)^3 \sin^2\theta_{k_1}
\Big]\, d_{1,k}^{-\beta} d_{1,k'}^{-\beta} \, {\rm d} \theta_k \, {\rm d} d_{1,k} \, {\rm d} \theta_{k'} \, {\rm d} d_{1,k'}\Big\}\sum_{k=1}^K \sum_{k'=1}^K 2{\sigma_{k}^2\sigma_{k^{\prime}}^2}\Big/\{\rho\mathbb{E}[g_2({\bf v})^2]\}
\end{aligned}
\end{equation}
\end{small}
\end{figure*}

{\bf Example}: Assume that the angles and distances of the UEs are independent, and that the angles are uniformly distributed over $[0, 2\pi)$. Under this assumption, the expectation terms can be simplified to derive a closed-form polynomial gradient. {Note that while this independence assumption simplifies the derivation to reveal core physical insights, the resulting GWF strategy is applicable to general scenarios with coupled joint angle-distance distributions, as will be verified in Section V.} 

Under the assumption, {$\mathbb{E}[f\left(k, k^{\prime}\right)\left[f\left(k, k^{\prime}\right)-f\left(k^{\prime}, k\right)\right]]$ $\mathbb{E}[ v_k v_{k^{\prime}}]$} can be further simplified as
\begin{equation}\label{equ19}
\begin{aligned}
&{\mathbb{E}\{f\left(k, k^{\prime}\right)\left[f\left(k, k^{\prime}\right)-f\left(k^{\prime}, k\right)\right]\}\mathbb{E}[ v_k v_{k^{'}}]}\\
&\approx \displaystyle\int_{ 0 }^{ 2 \pi } \dfrac{ 1 }{ 2 \pi }\displaystyle\int_{ 0 }^{ 2 \pi} \dfrac{ 1 }{ 2 \pi } \Big[(1+ \cos ( x) ) \sin ( y) ((1+ \cos ( x)) \sin ( y) \\
&\ \ \ -(1+ \cos ( y) ) \sin ( x) )\Big] {\rm d} x {\rm d} y{\mathbb{E}[ v_k v_{k^{'}}] = 0.75\mathbb{E}[ v_k v_{k^{'}}]},
\end{aligned}
\end{equation}
{Thus, we have}
\begin{equation}
\begin{aligned}
{\mathbb{E}[\frac{\partial g_1({\bf v})}{\partial v_{k_1}} g_2({\bf v})]\approx \sum_{k=1}^K \sum_{k^{\prime}=1}^K1.5\left(1+\operatorname{cos}\theta_{k_1}\right)\mathbb{E}[ v_k v_{k^{'}}],}
\end{aligned}
\end{equation}
which describes the asymptotic behavior of the CRLB under uniform angular distribution. Furthermore, for this case, we have
\begin{equation}
\begin{aligned}
&{\mathbb{E}[1+\cos \theta_k] = 1,} \\
&{\mathbb{E}[(1+\cos \theta_{k^{'}})^2] = 1.5,}\\ &{\mathbb{E}[\sin ^2 \theta_{k^{'}}]=0.5,}\\
&{\mathbb{E}[(1+\cos \theta_{k^{'}}) \sin \theta_{k^{'}}]=0.}
\end{aligned}
\end{equation}
Thus, we have
\begin{small}
\begin{equation}
\begin{aligned}
&{\mathbb{E}[g_1({\bf v}) \frac{\partial {g}_2({\bf v})}{\partial v_{k_ 1}}]}\approx\sum_{k=1}^K \sum_{k^{'}=1}^K[3\sin ^2 \theta_{k_1}+(1+\cos\theta_{k_1})^2]{\mathbb{E}[v_k v_{k^{\prime}}]}.
\end{aligned}
\end{equation}\end{small}
Furthermore, we obtain
\begin{equation}\label{equ22}
\begin{aligned}
&{\mathbb{E}[\frac{\partial~ \textrm{CRLB}}{\partial v_{k_1}}]}\approx\\
&\frac{[2\cos^{2}\left(\theta_{k_1}\right)-0.5\cos\left(\theta_{k_1}\right)-2.5]\sum_{k=1}^K \sum_{k^{'}=1}^K{\mathbb{E}[v_k v_{k^{\prime}}]}}{\rho{\mathbb{E}[{g}_2({\bf v})^2]}},
\end{aligned}
\end{equation}
which represents the expected derivative of the CRLB with respect to $v_{k_1}$. Given that {$v_{k} = {d}_{1,k}^{-\beta}{d}_2^{-\beta} \sigma_{k}^2$}, the expected derivative of the CRLB with respect to {$\sigma_{k}^2$} can thus be formulated as
\begin{equation}\label{equ23}
\begin{aligned}
&{\mathbb{E}[\frac{\partial~ \textrm{CRLB}}{\partial \sigma_{k_1}^2}]=\mathbb{E}[\frac{\partial~ \textrm{CRLB}}{\partial v_{k_1}}\cdot\frac{\partial~ v_{k_1}}{\partial \sigma_{k_1}^2}]=\mathbb{E}[{d}_{1,k_{1}}^{-\beta}{d}_2^{-\beta}]}\\
&\cdot\frac{[2\cos^{2}\left(\theta_{k_1}\right)-0.5\cos\left(\theta_{k_1}\right)-2.5]\sum_{k=1}^K \sum_{k^{'}=1}^K{\mathbb{E}[v_k v_{k^{\prime}}]}}{\rho{\mathbb{E}[{g}_2({\bf v})^2]}}.
\end{aligned}
\end{equation}

According to (\ref{equ23}), since {$\mathbb{E}[v_k v_{k^{\prime}}]$} and {$\mathbb{E}\left[{g}_2^2({\bf v})\right]$} are both strictly positive for each UE, while the term $2\cos^{2}\left(\theta_{k_1}\right)-0.5\cos\left(\theta_{k_1}\right)-2.5$ is non-positive, it follows that {$\mathbb{E}\left[\frac{\partial~ \textrm{CRLB}}{\partial \sigma_{k_1}^2}\right]\le 0$} for all $k_1\in{1,\cdots,K}$. This result reveals that, for any UE, an increase in the frequency-domain aperture {$\sigma_{k}^2$} leads to a monotonic decrease in the CRLB, except for the degenerate case of {$\theta_{k_1}=\pi$} \footnote{If and only if {$\theta_{k_1}=\pi$}, the expected derivative of the CRLB with respect to {$\sigma_{k}^2$} equals zero.}. Consequently, the Proposition \ref{pro2_1} can be established.
\begin{proposition}\label{pro2_1}
\textit{Increasing the frequency-domain aperture of a UE whose angle {$\theta_{k_1}\ne\pi$} results in a decrease of CRLB, whereas increasing the frequency-domain aperture of a UE located at {$\theta_{k_1}=\pi$} has no effect on the CRLB.}
\end{proposition}

\subsubsection{Monotonicity and Steepest Descent Analysis}

Moreover, according to (\ref{equ23}), the expected CRLB with regard to the frequency-domain aperture of the $k$th UE differs from that of the $k_1$th UE only in the term {$\mathbb{E}[{d}_{1,{k_1}}^{-\beta}{d}_2^{-\beta}]$}$[2\cos^{2}\left(\theta_{k_1}\right)-0.5\cos\left(\theta_{k_1}\right)-2.5]$. This implies that, once the angles and distances of UEs are determined, the CRLB ratios between any two dimensions remain constant. Consequently, the direction of the gradient of {$\mathbb{E}[\textrm{CRLB}]$} is fixed across all feasible points {$\boldsymbol{\sigma}^2$} in the allocation space.
Specifically, the gradient of {$\mathbb{E}[\textrm{CRLB}]$} can be expressed as
\begin{equation}
\begin{aligned}
&\nabla {\mathbb{E}[\textrm{CRLB}]}=\frac{\sum_{k=1}^K \sum_{k^{'}=1}^K{\mathbb{E}[v_k v_{k^{\prime}}]}}{\rho{\mathbb{E}[{g}_2({\bf v})^2]}}\\
&\cdot
\begin{bmatrix}
{\mathbb{E}[{d}_{1,{1}}^{-\beta}{d}_2^{-\beta}]}[2\cos^{2}\left(\theta_{1}\right)\!-\!0.5\cos\left(\theta_{1}\right)\!-\!2.5]\\
\vdots\\
{\mathbb{E}[{d}_{1,{K}}^{-\beta}{d}_2^{-\beta}]}[2\cos^{2}\left(\theta_{K}\right)\!-\!0.5\cos\left(\theta_{K}\right)\!-\!2.5]
\end{bmatrix}.
\end{aligned}
\end{equation}

Furthermore, since {$\mathbb{E}[\textrm{CRLB}]$} decreases monotonically with {$\sigma_k^2$} for all $k$ except when {$\theta_k=\pi$}, the optimization problem of minimizing {$\mathbb{E}[\textrm{CRLB}]$} under a total frequency-aperture constraint can be viewed as a directional resource allocation problem. In particular, to achieve the most effective reduction in {$\mathbb{E}[\textrm{CRLB}]$}, the frequency-domain apertures should be distributed along the direction of steepest descent, which corresponds to the gradient direction {$\nabla \mathbb{E}[\textrm{CRLB}]$}.

\begin{proposition}\label{pro2}
\textit{Under a total frequency-aperture constraint, distributing the frequency-domain apertures of UEs proportionally to the direction of {$\nabla \mathbb{E}[\rm{CRLB}]$} minimizes the expected CRLB.}
\end{proposition}

{\bf Remark}:
This result establishes a sensing counterpart to the classic capacity-maximizing water-filling algorithm in communications.
In communication systems, power is allocated to channels with higher signal-to-noise ratios (channel gains).
Analogously, in our network-level ISAC system, the term {$\left| \mathbb{E}[\frac{\partial~ \textrm{CRLB}}{\partial \sigma_{k_1}^2}] \right|$} can be interpreted as the \textit{``geometric gain"} of the $k$-th transceiver.
Nodes with favorable geometry (e.g., providing orthogonal perspectives) exhibit higher sensitivity (larger gradients), akin to ``deep channels" , and are thus allocated more aperture resources. Conversely, nodes with poor geometry (e.g., collinear with the target) have vanishing gradients and receive negligible resources. We term this strategy \textit{geometric water-filling}, as it maximizes the marginal sensing gain of the network.

\section{VGPA: From Geometric Water-Filling to Resource Allocation}

While the geometric water-filling principle provides the optimal continuous aperture values $\{\sigma_{k,\text{opt}}^2\}_{k=1}^{K}$, mapping these values to discrete subcarrier sets $\{\mathcal{S}_k\}_{k=1}^{K}$ in an OFDMA system with constraints on communication rate is a non-trivial combinatorial problem. In this section, we formulate this challenge as an integer set partitioning problem and propose a low-complexity algorithm.

\begin{algorithm}[tbp]
\caption{Variance-Guided Partitioning Algorithm}
\label{alg:variance_partition}
\begin{algorithmic}[1]
\STATE \textbf{Input:} Integer set $\mathcal{N}=\{1,2,\dots,N\}$; target sizes $\{n_k\}_{k=1}^K$; 

\STATE \textbf{Step 1: Initialization guided by target variances}
    \STATE Obtain variance ratio $\boldsymbol{\alpha}$ by geometric water-filling strategy
    \STATE Compute total variance $\sigma^2_{\text{total}}$ and per-subset target variance $\sigma_{k,\textrm{target}}^2$ by (\ref{equ:38})
    \FOR{$k=1$ to $K$ (in descending order of $\sigma_{k,\textrm{target}}^2$)}
        \STATE Select $n_k$ elements from $\mathcal{N}$ with spacing proportional to $\sqrt{\sigma_k^2}$
        \STATE Assign to subset $\mathcal{S}_k$ and remove them from $\mathcal{N}$
    \ENDFOR
\vspace{0.4em}

\STATE \textbf{Step 2: Variance-Consistent Refinement}
    \FOR{$t=1$ to $T_{\max}$}
        \STATE Randomly select two subsets $\mathcal{S}_i$ and $\mathcal{S}_j$
        \STATE Exchange one element between them to obtain $\{\mathcal{S}'_i,\mathcal{S}'_j\}$
        \STATE Evaluate the change value of the objective function $\Delta$
        \IF{$\Delta < 0$}
            \STATE Accept the exchange
        \ELSE
            \STATE Accept with probability $p = \exp(-\Delta / T_t)$ \hfill according to the Metropolis criterion
        \ENDIF
        \STATE Optionally decrease $T_t$ gradually to enhance convergence
    \ENDFOR
\vspace{0.4em}

\STATE \textbf{Output:} $\{\mathcal{S}_k\}_{k=1}^K$ 
\end{algorithmic}
\end{algorithm}

\subsection{Problem Formulation}
The resource allocation task is to partition the integer set of subcarriers $\mathcal{N} = \{1, 2, \dots, N\}$ into $K$ disjoint subsets $\{\mathcal{S}_k\}_{k=1}^K$ satisfying three key constraints
\begin{enumerate}
    \item {Cardinality constraint:} $|\mathcal{S}_k| = N_k$, ensuring the total bandwidth budget is met.
    \item {Aperture matching:} the variance of each subset, $\text{Var}(\mathcal{S}_k)$, should match the optimal target variance $\sigma_{k,\text{opt}}^2$ derived from the GWF principle.
    \item {Mean consistency:} the mean of each subset, $\text{Mean}(\mathcal{S}_k)$, should align with the global center frequency $\mu = (N+1)/2$ to maximize the achievable frequency aperture (refer to Proposition 4 in \cite{wang2024fundamentaltradeoffstimefrequencyresource}).
    \item {Communication rate constraint:} the minimized communication rate of each UE shuld be satisfied.
\end{enumerate}

Mathematically, this can be formulated as minimizing the deviation from these targets
\begin{small}
\begin{equation} \label{eq:obj_func}
\begin{aligned}
\min_{\{\mathcal{S}_k\}} \quad & \sum_{k=1}^K (\mu_k - \mu)^2 + \sum_{k=1}^K \left( \frac{\sigma_k^2}{\alpha_k} - \frac{1}{K} \sum_{j=1}^K \frac{\sigma_j^2}{\alpha_j} \right)^2 \\
\textrm{s.t.} \quad & \bigcup_{k=1}^K \mathcal{S}_k = \mathcal{N}, \quad \mathcal{S}_i \cap \mathcal{S}_j = \emptyset, \quad |\mathcal{S}_k| = N_k,\\
& \sum_{n\in\mathcal{S}_{k}}\log_{2}(1+\mathrm{SNR}_{k,n})\geq R_{k}^{\mathrm{req}},\quad\forall k\in\{1,\ldots,K\}
\end{aligned}
\end{equation}\end{small}
where $\mu_k$ and $\sigma_k^2$ are the mean and variance of subset $\mathcal{S}_k$, respectively. $\alpha_k$ is the target aperture ratio, $\mathrm{SNR}_{k,n}$ is the signal-to-noise (SNR) ratio of the $n$th subcarrier of the $k$th UE, and $R_{k}^{\mathrm{req}}$ is defined as the minimized rate required by the $k$th UE.
Finding the global optimum for (\ref{eq:obj_func}) is an NP-hard integer partitioning problem. Exhaustive search is computationally prohibitive for practical values of $N$ (e.g., $N \ge 64$).

\subsection{Theoretical Insight: The Variance-Stride Relationship}
To design an efficient heuristic, we look for a constructive method to generate subsets with specific variances. We exploit a fundamental geometric property of discrete sequences: \textit{the variance is dominated by the spacing (stride) between elements.}

\begin{proposition}\label{pro4}
For a discrete set of $N_k$ integer elements selected with a uniform stride (spacing) $\eta_k$, the variance of the selected subset is proportional to $\eta_k^2$. 
\end{proposition}
\textit{Proof:} See Appendix \ref{proof2}.

\textbf{Remark:} Proposition \ref{pro4} implies that controlling the \textit{stride} provides a direct and effective mechanism to control the \textit{frequency aperture}. A subset requiring a large variance (high geometric gain) should be allocated subcarriers with a large stride (i.e., a sparse comb-like pattern), while a subset with low variance requirement should be allocated contiguous subcarriers (small stride). This insight forms the basis of our proposed algorithm.

\subsection{The Variance-Guided Partitioning Algorithm (VGPA)}
Based on the variance-stride relationship, we propose the VGPA, which operates in two stages: a constructive initialization and a stochastic refinement.

\subsubsection{Stage 1 (Variance-Guided Initialization)}
This stage constructs a high-quality initial partition by greedily allocating subcarriers based on the required stride.
First, we calculate the target variance $\sigma_{k,\text{target}}^2$ for each user based on the global variance of $\mathcal{N}$ and the target ratios $\alpha_k$ (see Appendix \ref{proof3} for derivation): 
\begin{equation}\label{equ:38}
\sigma_{k,\text{target}}^{2} = \frac{N\alpha_{k}\sigma_{\text{total}}^2}{\sum_{k=1}^{K}N_k\alpha_k}.
\end{equation}
Then, we employ a \textit{largest-aperture-first} strategy: users are sorted in descending order of their target variances. For each user $k$, we estimate the required stride $\eta_k \propto \sqrt{\sigma_{k,\text{target}}^2}$. We then select $N_k$ available subcarriers from $\mathcal{N}$ using this stride, centering the selection around the global mean $\mu$ to satisfy the mean consistency constraint simultaneously. This constructive approach avoids collision and ensures the initial partition is already close to the optimal geometric configuration.

\subsubsection{Stage 2 (Rate-Constrained Iterative Refinement )}
While initialization provides a geometrically favorable partition, it may not strictly satisfy the communication rate constraints defined in (\ref{eq:obj_func}). To address this, we reformulate the optimization problem using the penalty method. Specifically, we incorporate a quadratic penalty term into the objective function:
\begin{equation}
\begin{small}
\begin{aligned}
&\sum_{k=1}^K \Big[(\mu_k - \mu)^2 +  \left( \frac{\sigma_k^2}{\alpha_k} - \frac{1}{K}  \frac{\sigma_j^2}{\alpha_j} \right)^2\\
&+\max\left(0,R_k^{req}-\sum_{n\in\mathcal{S}_{k}}\log_{2}(1+\mathrm{SNR}_{k,n})\right)^2\Big].
\end{aligned}
\end{small}
\end{equation}
We then refine the partition using a stochastic local search. In each iteration, we randomly select two subsets and swap an element between them. The swap is accepted based on the Metropolis criterion to escape local minima. According to \cite{4767596}, the refinement process following Metropolis criterion asymptotically converges to the global optimum with probability one under a sufficiently slow cooling schedule \cite{4767596}.
This implicitly enforces the \textit{geometric water-filling} principle, ensuring that the critical aperture ratios are strictly met, while mean consistency serves as a secondary soft constraint. The detailed procedure is summarized in {\bf Algorithm} \ref{alg:variance_partition}, where $T_\textrm{max}$ is the maximum number of iterations.

\subsection{Complexity Analysis}
The complexity of VGPA is dominated by the two stages. The initialization traverses the subcarrier set $\mathcal{N}$ once, scaling linearly with $\mathcal{O}(N)$. The refinement stage runs for a fixed number of iterations $I$, with each objective evaluation taking $\mathcal{O}(K)$. Thus, the total complexity is $\mathcal{O}(N + IK)$. Since the complexity scales linearly with the number of subcarriers $N$ rather than exponentially, VGPA ensures high scalability for practical wideband ISAC systems.

\section{Numerical Simulation}
In this section, numerical simulations are performed to evaluate our theoretical analysis. 
Specifically, we set the carrier frequency $f_\textrm{c}$ to be $28$ GHz, while the subcarrier spacing is chosen to be $100$ kHz \cite{rahman2019framework, wxy}. 
Then, the base station is fixed at coordinates $[0,70]$ meters, and the sensing target is at the origin $[0,0]$ meters. Unless otherwise specified, the positions of the distributed transmitters are randomly generated within a circular region with a radius of 40 meters. For simplicity, simulations are conducted with 50 distributed transmitters, as depicted in Fig. \ref{figure10}.








\begin{figure}[tbp]
	\centering
	\includegraphics[width=6.9cm]{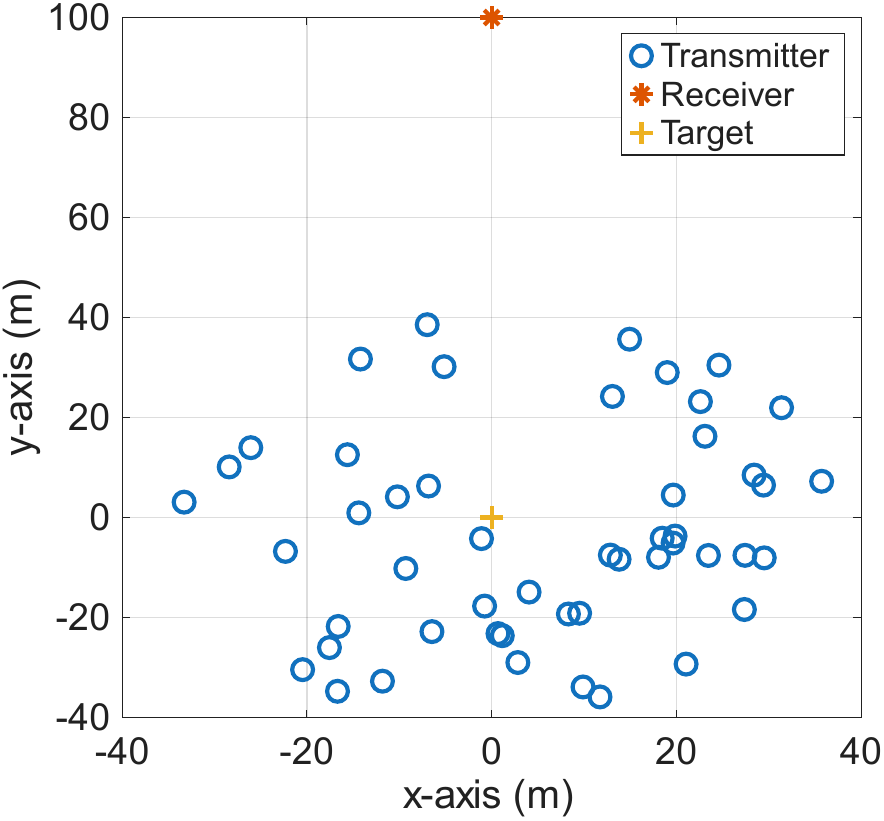}
	\caption{Network-level distributed ISAC set-up.}
	\label{figure10}
\end{figure}

\begin{figure}[tbp]
	\centering
	\includegraphics[width=7.4cm]{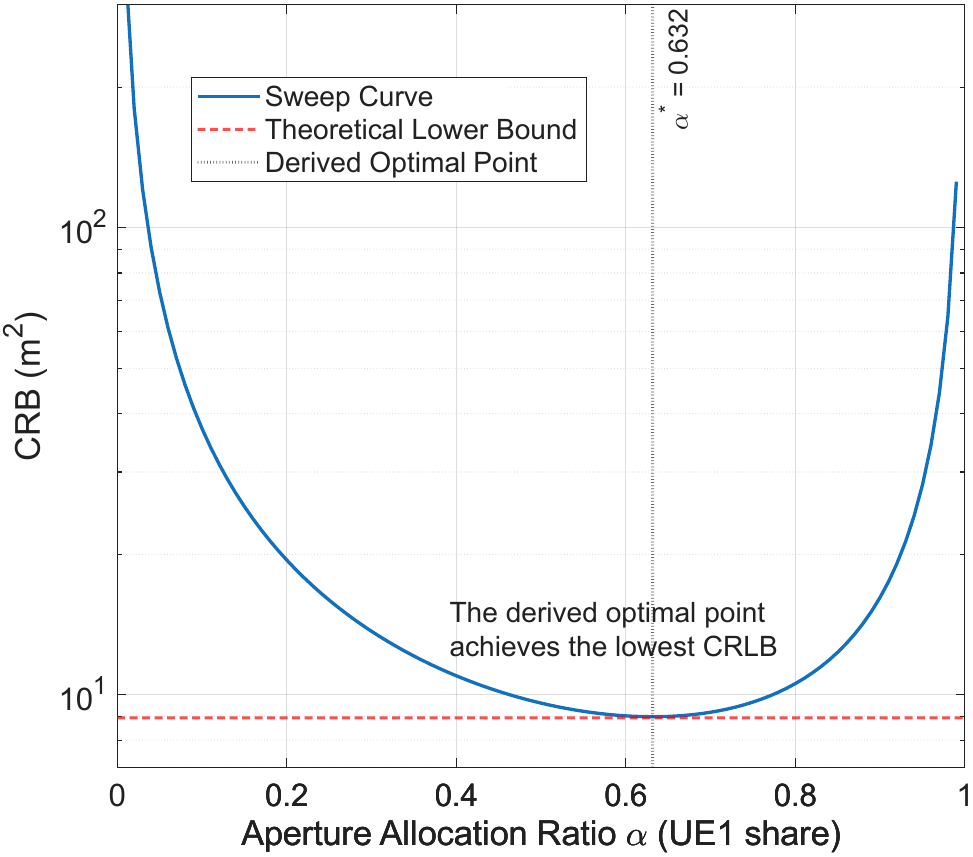}
	\caption{The CRLB corresponding to different aperture allocation scheme in two-UE scenarios.}
	\label{figure11}
\end{figure}

Fig. \ref{figure11} illustrates the CRLB performance as a function of the aperture-allocation ratio~$\alpha$ between the two UEs. In this simulation, we set two UEs at $(50,40)$m and $(-30,-30)$m with $64$ and $64$ subcarriers, respectively. As shown in the figure, the CRLB exhibits a symmetric bowl-shaped trend, becoming extremely large when $\alpha$ approaches $0$ or $1$, i.e., when one UE occupies almost the entire aperture budget while the other receives negligible sensing resources. This confirms the intuitive fact that assigning all aperture to a single UE severely degrades cooperative sensing performance.
When $\alpha = 0.607$, the CRLB curve reaches its minimum and matches the analytical lower bound, validating both the correctness and the tightness of our theoretical derivation.


\begin{figure}[tbp]
    \centering
    \subfloat[]{
        \includegraphics[width=0.5\linewidth,height=0.74\linewidth]{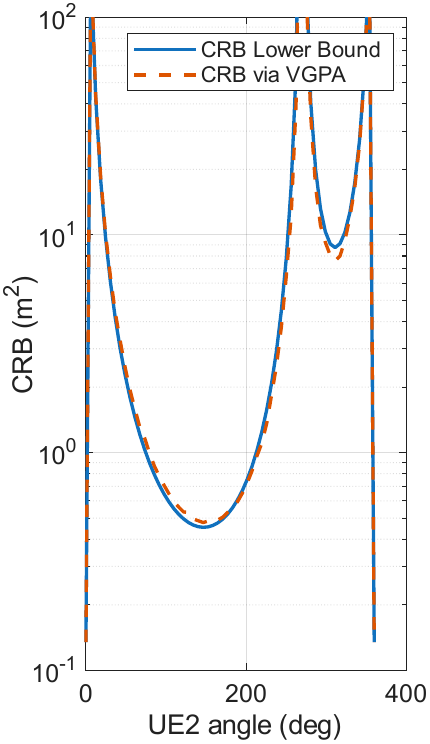} 
        \label{fig:angle_scan}
    }
    \subfloat[]{
        \includegraphics[width=0.47\linewidth,height=0.73\linewidth]
        {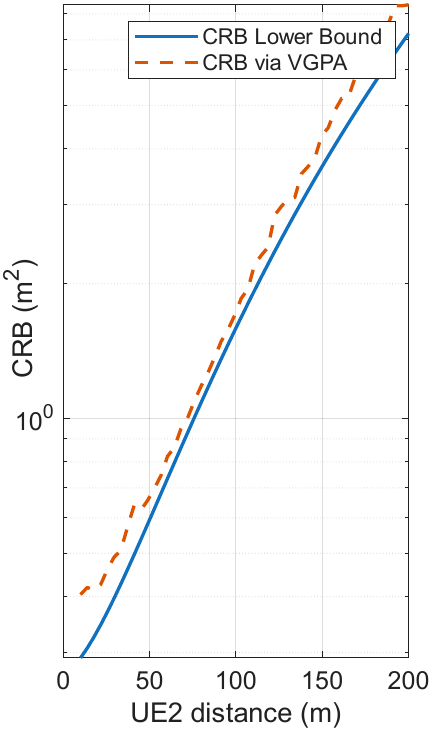} 
        \label{fig:dist_scan}
    }
    \caption{The CRLB lower bound and CRLB via VGPA with different (a) angles and (b) distances of the $2$nd UE in two-UE scenarios.}
    \label{figure12}
\end{figure}

Fig. \ref{figure12} compares the theoretical CRLB lower bound with the CRLB obtained via the proposed VGPA-based aperture allocation under two different geometric variations of the second UE.
In the first subfigure, the angle of the first UE is fixed at $\theta_1=0$, and the distance of the second UE is set to $60$~m, while $\theta_2$ varies from $0^\circ$ to $360^\circ$.
The results explicitly visualize the impact of sensing geometry on the validity of the optimization framework.
Specifically, we observe two distinct regimes where the localization performance degrades severely due to the \textit{rank deficiency} of the FIM:
\begin{itemize}
    \item {Geometric redundancy ($\theta_2 \approx \theta_1$):} When the two UEs are nearly co-located (around $0^\circ$ and $360^\circ$), the geometric steering vectors $\mathbf{u}_1$ and $\mathbf{u}_2$ become linearly dependent, the FIM becomes ill-conditioned, and the cooperative gain vanishes.
    \item {Geometric singularity ($\theta_2 \approx \theta_0 + \pi$):} when $\theta_2 \approx 270^\circ$, namely a transmitter aligns with the target-BS axis in the backward direction, its corresponding FIM eigenvalues approach zero. This reduces the effective rank of the FIM from 2 to 1, rendering 2D localization mathematically impossible.
\end{itemize}
Across the valid angular range, the VGPA-based CRLB closely tracks the theoretical lower bound, confirming that the proposed algorithm effectively exploits the sensing geometry to minimize the CRLB.
The second subfigure shows the CRLB as a function of the distance of UE2, with its angle fixed at $\pi/2$ (an orthogonal, well-conditioned geometry).
As expected, the CRLB increases monotonically with distance due to path loss.
In this non-singular regime, the VGPA-generated CRLB remains consistently close to the theoretical lower bound, validating the accuracy of the proposed aperture allocation scheme in general cooperative ISAC scenarios.

\begin{figure}[tbp]
	\centering
	\includegraphics[width=6.8cm]{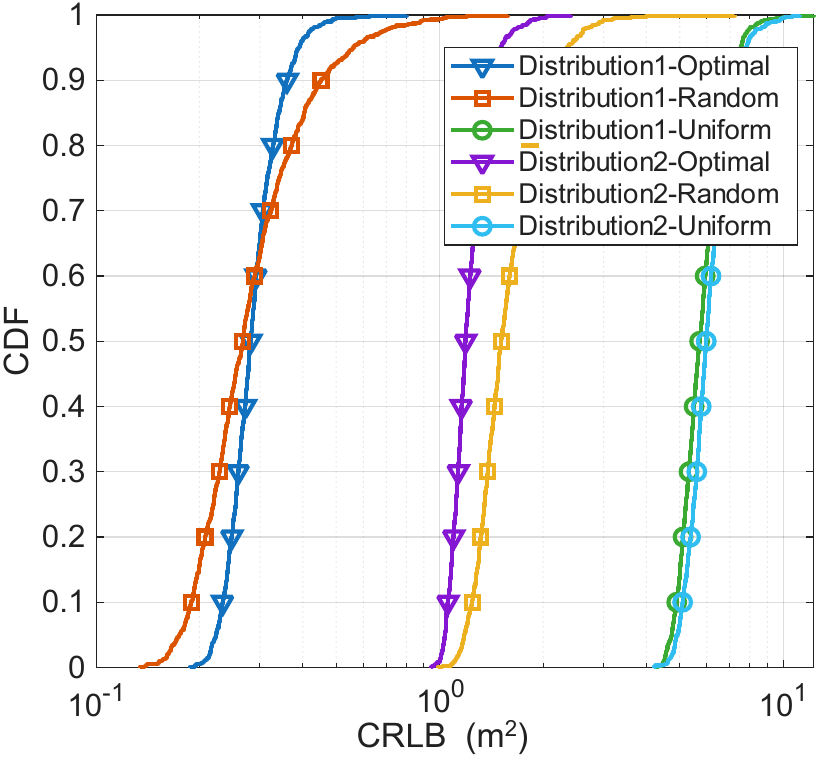}
	\caption{CDF of CRLB with different tranceivers distribution.}
	\label{figure13}
\end{figure}

\begin{figure}[tbp]
    \subfloat[]{
        \includegraphics[width=7.7cm]{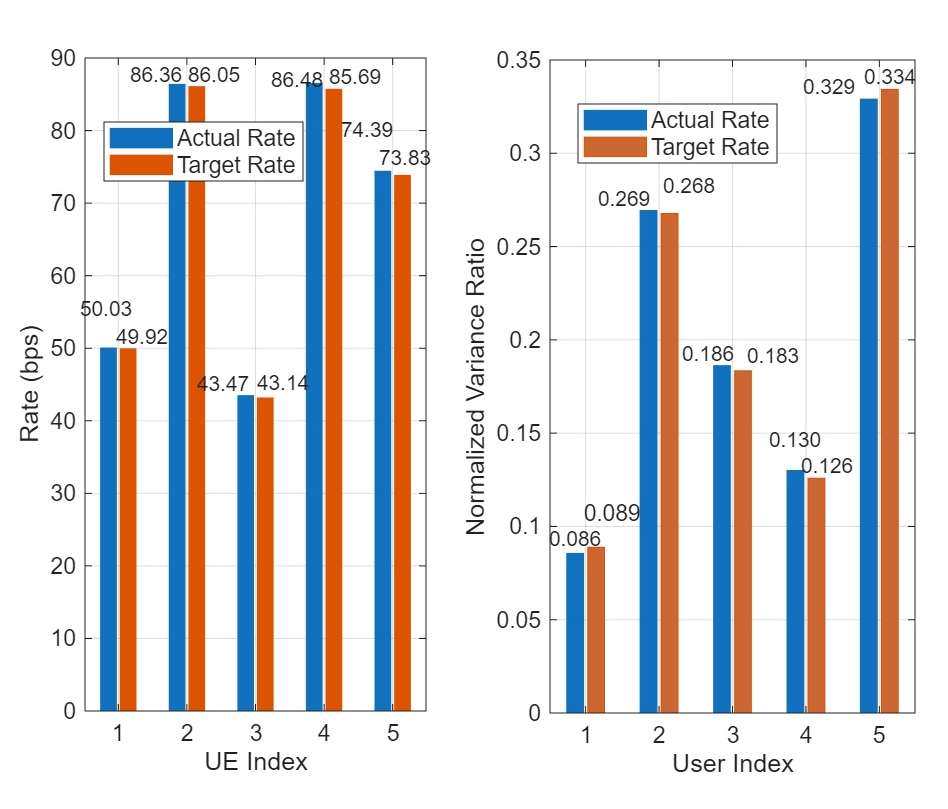}
        \label{figure14}
    }
    \hfill
    \subfloat[]{
    \hspace{0.5cm}
        \includegraphics[width=6.8cm]{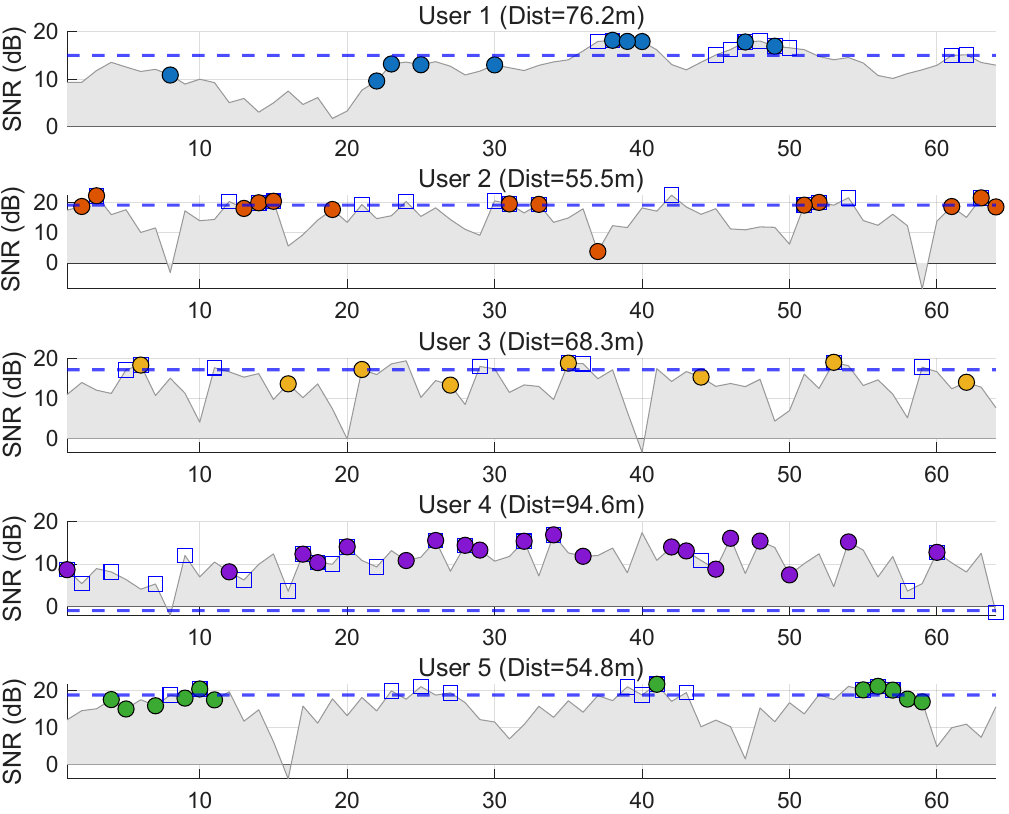}
        \label{figure15}
    }
    \caption{Verification of resource allocation in five-UE scenarios under communication rate constraints. (a) Comparison between the theoretical optimal variance ratios/target rates and the actual values achieved by VGPA. (b) Comparison of subcarrier allocation patterns between VGPA and traditional water-filling. The blue squares represent subcarriers selected by water-filling, while the colorful circles represent those selected by the proposed VGPA. The background grey area indicates the SNR profile of the subcarriers.}
    \label{figure_combined1}
\end{figure}

Fig. \ref{figure13} illustrates the CDF of the sensing CRLB under two different transmitter–receiver spatial distributions and three aperture-allocation schemes: Optimal (GWF), Random, and Uniform. For Distribution 2, where transceivers are uniformly distributed over the full angular range $(0, 2\pi)$, the proposed Optimal allocation achieves the best performance. Interestingly, the Uniform allocation, which assigns equal aperture resources to all UEs, yields the worst performance, even underperforming the Random baseline. This counter-intuitive result highlights the severity of the geometric coupling: Uniform allocation rigidly forces resources into nodes with poor sensing geometry (e.g., those collinear with the target-receiver axis or in deep fades), thereby inefficiently diluting the total aperture budget. In contrast, the Random allocation introduces resource diversity; although geometry-agnostic, its stochastic nature allows for unequal distributions that can probabilistically assign more resources to high-gain nodes than the rigid uniform scheme. For Distribution 1, where transceivers are uniformly placed within $(0, 1.5\pi)$, the overall CRLB is lower than that of Distribution 2 due to the concentrated deployment. However, the performance hierarchy remains consistent, with the Optimal allocation maintaining the lowest CRLB by explicitly exploiting geometric sensitivity to perform ``water-filling.'' Moreover, in this case, the performance gap between optimal and random allocation is small, with the mean CRLB under the optimal scheme being only $0.0115$ m smaller than that under the random allocation.
These results indicate that (i) the spatial distribution of transmitters significantly affects cooperative sensing performance, and (ii) the proposed optimal allocation consistently outperforms the random and uniform baseline under both wide-angle and sector-limited deployment scenarios.



Fig. \ref{figure_combined1} validates the effectiveness of the proposed VGPA in scenarios with strict communication rate constraints. Fig. 6(a) quantitatively compares the target metrics against the actual results achieved by the algorithm. It is observed that the actual communication rates for all UEs strictly satisfy the minimum rate requirements, validating the efficacy of the penalty-based formulation. Simultaneously, the actual variance ratios achieved by VGPA closely track the theoretical optimal ratios derived from Proposition 3. This confirms that VGPA successfully navigates the trade-off, satisfying communication QoS while precisely shaping the resource distribution to match the sensing geometry requirements.Fig. 6(b) provides a microscopic view of the fundamental conflict and trade-off between sensing geometry and communication channels. The grey background represents the frequency-selective fading channel (SNR) for specific users. The blue squares illustrate the allocation by the traditional capacity-maximizing Water-Filling algorithm, while the colorful circles represent the proposed ISAC allocation.As observed, the Water-Filling strategy (blue squares) strictly occupies the 'channel peaks' to maximize SNR, often resulting in clustered subcarriers with small frequency variance (narrow aperture). In contrast, the proposed VGPA (circles) adopts a 'channel-aware aperture expansion' strategy: while it still favors high-SNR regions to meet rate constraints, it actively disperses subcarriers towards the band edges to maximize frequency variance. This visualization effectively demonstrates that the proposed scheme intelligently sacrifices marginal SNR gains—by selecting subcarriers that are 'good enough' rather than 'optimal' in channel gain—to achieve significant gains in sensing aperture.

\begin{figure}[tbp]
    \centering
    \subfloat[]{
        \includegraphics[height=6.8cm,width=7.3cm]{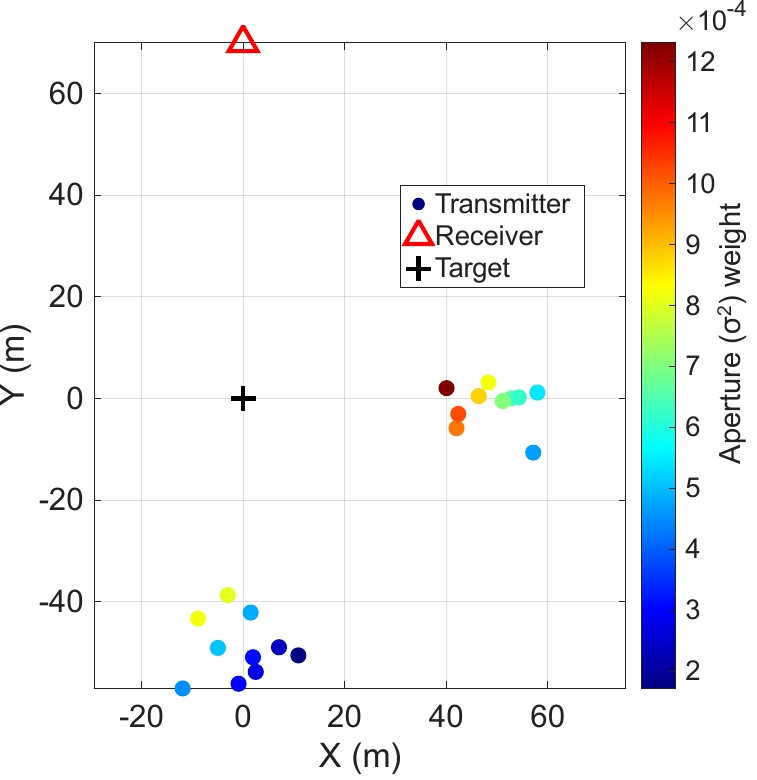}
        \hspace{-0.9cm}
        \label{figure141}
    }
    \hfill
    \subfloat[]{
    \hspace{-0.3cm}
        \includegraphics[width=6.6cm]{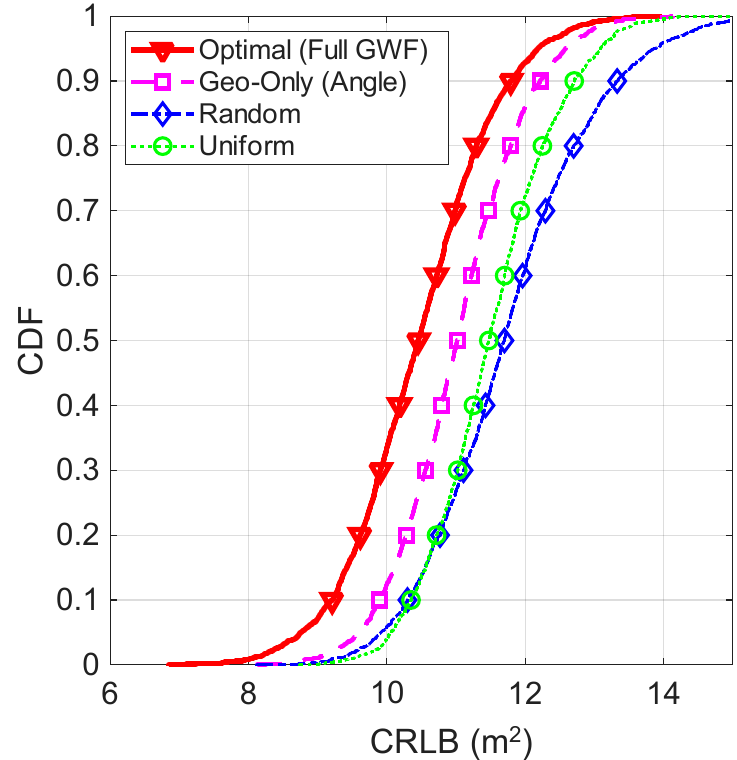}
        \label{figure151}
    }
    \caption{Evaluation under a non-typical Gaussian mixture distribution. (a) The spatial distribution of transmitters forms two distinct clusters: a ``good geometry" cluster (orthogonal to receiver-Target axis) and a ``bad geometry" cluster (collinear). The color intensity represents the aperture weight allocated by the proposed strategy. (b) CDF comparison of sensing CRLB.}
    \label{figure_combined}
\end{figure}

\begin{figure}[tbp]
    \centering
    \subfloat[]{
    \hspace{0.5cm}
    \vspace{1cm}
    \includegraphics[height=6.8cm]{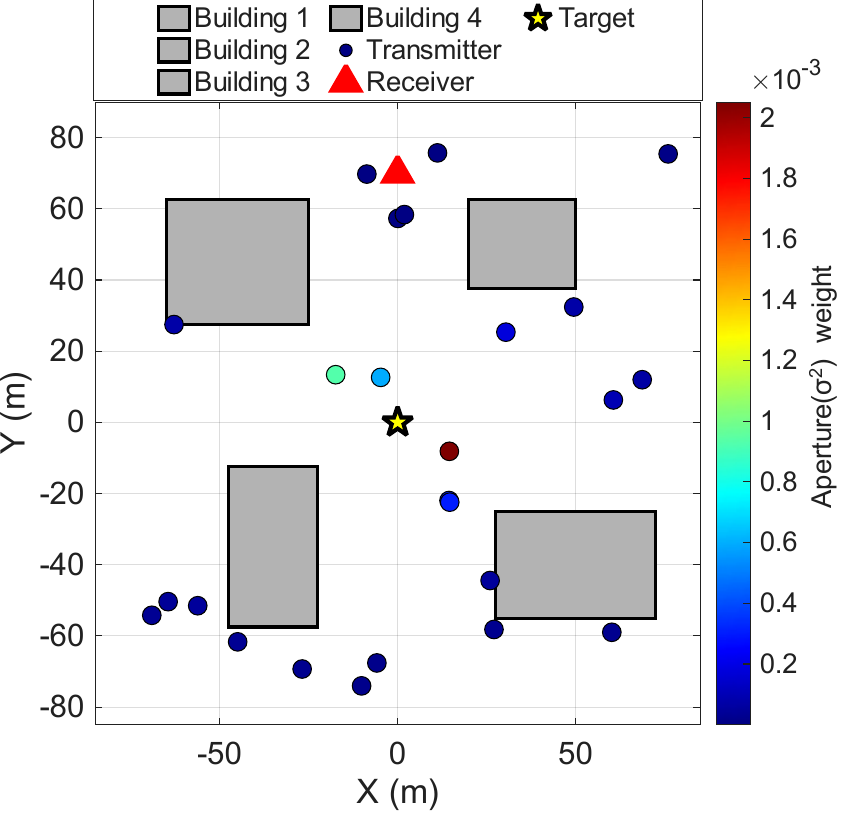}
        \label{figure142}
    }
    \hfill
    \subfloat[]{

        \includegraphics[width=6.8cm]{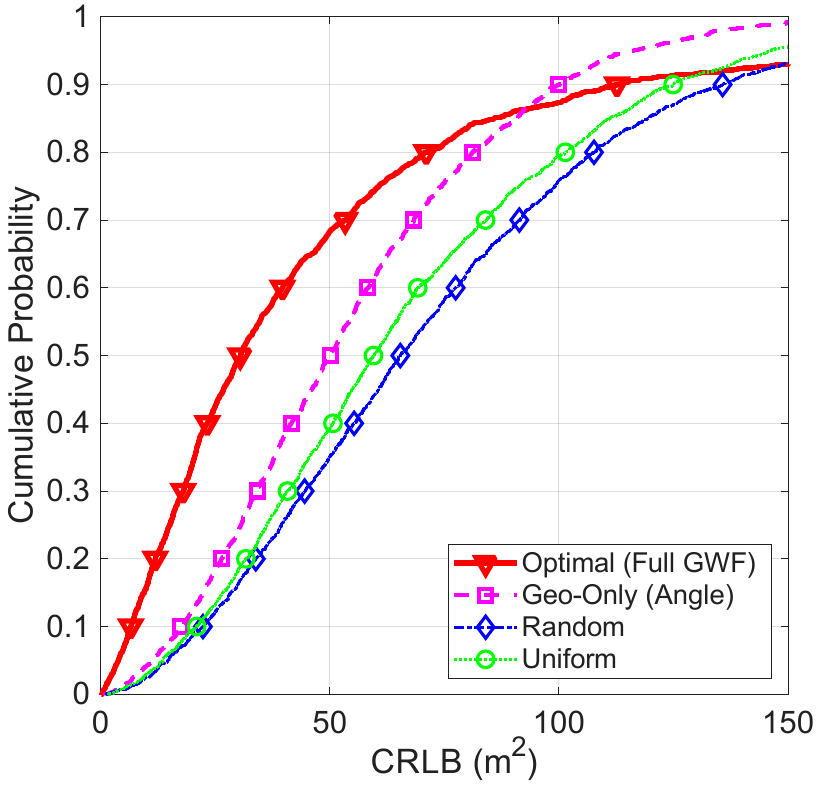}
        \label{figure152}
    }
    \caption{Performance in an asymmetric urban street canyon scenario. (a) Simulation setup where UEs are physically constrained to streets by buildings (grey blocks). (b) CDF of Sensing CRLB, highlighting the performance degradation of the Uniform strategy due to resource dilution in geometric blind zones.}
    \label{figure_combined-1}
\end{figure}

In Fig. \ref{figure141}, we verify the adaptability of the ``geometric water-filling" principle in a a non-typical joint angle-distance distribution involving two Gaussian hotspots, as shown . Hotspot 1 is located at the orthogonal position relative to the Receiver-Target axis, offering high geometric gain, while Hotspot 2 is located in the backscatter region (collinear), offering negligible angular information. The visualization reveals that the proposed algorithm functions as a ``geometric filter": it assigns high aperture weights (represented by red/yellow colors) to UEs in Hotspot 1 while suppressing UEs in Hotspot 2 (blue colors). Fig. \ref{figure151} confirms that the Optimal scheme achieves the lowest CRLB, whereas the Uniform and Random schemes suffer performance losses by wasting resources on the geometrically inefficient Hotspot 2.

In {Fig. \ref{figure_combined-1}}, we evaluate the proposed framework in a realistic asymmetric urban street canyon environment, where the joint distribution of angle and distance is physically constrained by buildings, as illustrated in {Fig. \ref{figure142}}. The sensing CRLB results in {Fig. \ref{figure152}} highlight a critical performance hierarchy. The {Random} allocation yields the worst performance due to the sparsity of effective nodes in street canyons, where stochastic allocation often fails to capture critical geometric positions. The {Uniform} strategy ranks third; while providing a ``safety net'' against worst-case random assignments, it inefficiently dilutes resources into geometric blind zones (e.g., vertical streets). Notably, the {Geo-Only} strategy significantly outperforms both geometry-agnostic baselines, effectively bridging the majority of the performance gap between the baselines and the {Optimal} scheme. This indicates that while path loss information contributes to the final precision gain (as evidenced by the gap between Optimal and Geo-Only), prioritizing favorable angular topologies is the fundamental prerequisite for avoiding severe performance degradation in constrained environments. Consequently, by jointly exploiting both geometric sensitivity and channel gain, the proposed GWF principle provides the most robust solution. Crucially, these results validate that the proposed GWF framework remains robust even under realistic constraints where angle and distance distributions are highly coupled (e.g., street canyons), successfully generalizing beyond the independent distribution assumptions used in the theoretical derivation.

\section{Conclusion}
This paper has investigated the fundamental coupling between the spatial geometry of distributed transceivers and the allocation of time-frequency resources in network-level ISAC systems. 
First, for the foundational two-transmitter system, we analytically derived the closed-form optimal resource allocation to minimize the CRLB. 
Second, extending to the general multi-transmitter scenario, we analytically derived the optimal aperture allocation and revealed a fundamental \textit{geometric water-filling} principle. This principle dictates that resources should be preferentially allocated to nodes with superior sensing geometry, providing a clear physical guideline for system design.
Third, to bridge the gap to practical implementation, we formulated the subcarrier assignment with constraints on communication rate as an NP-hard combinatorial problem and proposed a low-complexity Variance-Guided Partitioning Algorithm (VGPA) that yields near-optimal solutions.
Numerical results consistently corroborate the theoretical findings.


\appendices
\section{Proof of Proposition 1}
Since the optimization objective is convex,  (\ref{equ10-2}) can be rewritten in a simplified form
\begin{equation}
\begin{aligned}
\min_{\sigma_1^2, \sigma_2^2} \quad & \frac{\Gamma_1}{\sigma_1^2} + \frac{\Gamma_2}{\sigma_2^2} \\
\text{s.t.} \quad & N_1 \sigma_1^2 + N_2 \sigma_2^2 = E_{\text{total}}.
\end{aligned}
\end{equation}
Then we apply the Lagrange multiplier method and the Lagrangian is
\begin{equation}
\mathcal{L} = \frac{\Gamma_1}{\sigma_1^2} + \frac{\Gamma_2}{\sigma_2^2} + \lambda (N_1 \sigma_1^2 + N_2 \sigma_2^2 - E_{\text{total}}).
\end{equation}
Taking the partial derivatives with respect to $\sigma_1^2$ and $\sigma_2^2$ and setting them to zero, we have
\begin{align}
\frac{\partial \mathcal{L}}{\partial \sigma_1^2} &= -\frac{\Gamma_1}{(\sigma_1^2)^2} + \lambda N_1 = 0 \Rightarrow \sigma_1^2 = \sqrt{\frac{\Gamma_1}{\lambda N_1}}, \label{eq:dev1} \\
\frac{\partial \mathcal{L}}{\partial \sigma_2^2} &= -\frac{\Gamma_2}{(\sigma_2^2)^2} + \lambda N_2 = 0 \Rightarrow \sigma_2^2 = \sqrt{\frac{\Gamma_2}{\lambda N_2}}. \label{eq:dev2}
\end{align}
From (\ref{eq:dev1}) and (\ref{eq:dev2}), the optimal ratio between the apertures is
\begin{equation}
\frac{\sigma_1^2}{\sigma_2^2} = \sqrt{\frac{\Gamma_1 N_2}{\Gamma_2 N_1}}.
\end{equation}
Substituting $\sigma_2^2 = \sigma_1^2 \sqrt{\frac{\Gamma_2 N_1}{\Gamma_1 N_2}}$ into the total energy constraint
\begin{equation}
N_1 \sigma_1^2 + N_2 \left( \sigma_1^2 \sqrt{\frac{\Gamma_2 N_1}{\Gamma_1 N_2}} \right) = E_{\text{total}}.
\end{equation}
Solving for $\sigma_1^2$, we obtain
\begin{equation}
\sigma_1^2 \left( N_1 + \sqrt{N_1 N_2 \frac{\Gamma_2}{\Gamma_1}} \right) = E_{\text{total}},
\end{equation}
which leads to the closed-form solution in (\ref{opt_sig1}). The derivation for $\sigma_2^2$ follows symmetrically. \qed

\section{Proof of Proposition \ref{pro4}}\label{proof2}

The ordered integer set {$\mathcal{N} = \{1, 2, \dots, N\}$} is to be partitioned into $K$ disjoint subsets {$\{\mathcal{S}_k\}_{k=1}^K$}, where each subset {$\mathcal{S}_k$} contains $N_k$ elements such that $\sum_{k=1}^K N_k = N$. Moreover, each subset is intended to achive a target variance $\alpha_k$.

During the initialization, we assume that the elements of {$\mathcal{S}_k$} are selected at equal intervals $\eta_k$.
Accordingly, the elements in {$\mathcal{S}_k$} can be expressed as
\begin{equation}
{\mathcal{S}_k} = \left\{ x_k + i \cdot \eta_k \,\big|\, i = 0, 1, \dots, N_k - 1 \right\},
\end{equation}
where $x_k$ denotes the starting index of subset {$\mathcal{S}_k$} within the ordered set.

Therefore, the variance of the sequence is given by
\begin{equation}
\begin{aligned}
\text{Var}({\mathcal{S}_k}) &= [0+\eta_k+\cdots+(N_K-1)\eta_k]^2/N_k\\ &=\frac{(N_k - 1)(2N_k - 1)}{6} \cdot \eta_k^2.
\end{aligned}
\end{equation}
Obviously, we have $\eta_k\propto \sqrt{\text{Var}({\mathcal{S}_k})}$. This thus completes the proof.

\section{Proof of Equation (\ref{equ:38})}\label{proof3}

According to \cite{scheffe1999analysis}, the variance of a complete set can be decomposed into the sum of the within-subset variance and the between-subset variance 
\begin{equation}
\sum\limits_{n=1}^{N}\!(n-\bar{n})^2\!\!=\!\sum\limits_{k=1}^{K}\sum\limits_{i=1}^{N_k}(n_{k,i}-\bar{n}_k)^2+\sum\limits_{k=1}^{K}\!N_k(\bar{n}_k-\bar{n})^2,
\end{equation} %
where $\bar{n}$ is the mean of all subcarrier indices, $n_{k,i}$ is the $i$th subcarrier index of the $k$th UE, and $\bar{n}_k$ is the mean of the 
$k$th subset.

To optimize sensing performance, the subset means should be aligned as closely as possible. For simplicity, we set $\bar{n}_k={\bar n}$ for $k=1,\cdots,K$, which eliminates the between-subset variance and yields
\begin{equation}\label{equ56}
\sum\limits_{n=1}^{N}\!(n-\bar{n})^2\!\!=\!\sum\limits_{k=1}^{K}\sum\limits_{i=1}^{N_k}(n_{k,i}-\bar{n}_k)^2=\sum\limits_{k=1}^{K}N_k\sigma_k^2=N\sigma_\textrm{total}^2,
\end{equation}

Given the target variance ratio vector  $[\alpha_1,\cdots,\alpha_K]$ across subsets (\ref{equ56}) can be rewritten as
\begin{equation}
N\sigma_\textrm{total}^2=\sum\limits_{k=1}^{K}N_k\alpha_k\sigma_1^2/\alpha_1,
\end{equation}
from which the variance of any target subset can be obtained as
\begin{equation}
\sigma_k^2=\frac{N\alpha_k\sigma_\textrm{total}^2}{\sum_{k=1}^{K}N_k\alpha_k}.
\end{equation}

This completes the proof.


\bibliographystyle{IEEEtran}
\bibliography{IEEEabrv,reference.bib}

@ARTICLE{10918629,
  author={Valiulahi, Iman and Masouros, Christos and Alaaeldin, Mahmoud and Alsusa, Emad},
  journal={IEEE Open Journal of Vehicular Technology}, 
  title={{ISAC} Receiver Design: Joint DoA and Data Estimation in the Presence of Incomplete Signal Observations}, 
  year={2025},
  volume={6},
  number={},
  pages={846-852}
}

@misc{valiulahi,
      title={{ISAC} Super-Resolution Receiver via Lifted Atomic Norm Minimization}, 
      author={Iman Valiulahi and Christos Masouros and Athina P. Petropulu},
      year={2024},
      eprint={2411.09495},
      archivePrefix={arXiv},
      primaryClass={cs.IT},
      url={https://arxiv.org/abs/2411.09495}, 
}

@ARTICLE{10726912,
  author={Meng, Kaitao and Masouros, Christos and Petropulu, Athina P. and Hanzo, Lajos},
  journal={IEEE Wireless Communications}, 
  title={Cooperative {ISAC} Networks: Opportunities and Challenges}, 
  year={2025},
  month={Jun.},
  volume={32},
  number={3},
  pages={212-219}
}

@ARTICLE{11098638,
  author={Song, Yuxuan and Zeng, Yong and Yang, Yuhang and Ren, Zixiang and Cheng, Gaoyuan and Xu, Xiaoli and Xu, Jie and Jin, Shi and Zhang, Rui},
  journal={IEEE Communications Magazine}, 
  title={An Overview of Cellular {ISAC} for Low-Altitude {UAV}: New Opportunities and Challenges}, 
  year={2025, early access},
  month={Jul.},
  volume={},
  number={},
  pages={1-8}
}

@ARTICLE{11205178,
  author={Duan, Yu-Ru and Yang, Shaoshi and Zhai, Hou-Yu and Wang, Xiao-Yang and Tan, Jing-Sheng and Luo, Yu-Song and Chen, Sheng},
  journal={IEEE Wireless Communications Letters}, 
  title={{WiIID}: {Wi-Fi} Based Intelligent Indoor Intrusion Detection With Tensor Decomposition}, 
  year={2025, early access},
month={Oct.},
  volume={},
  number={},
  pages={1-1}
}

@ARTICLE{11184506,
  author={Han, Kawon and Meng, Kaitao and Wang, Xiao-Yang and Masouros, Christos},
  journal={IEEE Journal of Selected Topics in Electromagnetics, Antennas and Propagation}, 
  title={Network-Level {ISAC} Design: State-of-the-Art, Challenges, and Opportunities}, 
  year={2025},
month={Sep.},
  volume={1},
  number={1},
  pages={65-83}
}

@ARTICLE{4767596,
  author={Geman, Stuart and Geman, Donald},
  journal={IEEE Transactions on Pattern Analysis and Machine Intelligence}, 
  title={Stochastic Relaxation, Gibbs Distributions, and the Bayesian Restoration of Images}, 
  year={1984},
  month={Nov},
  volume={PAMI-6},
  number={6},
  pages={721-741},
}

@ARTICLE{meng2024,
  author={Meng, Kaitao and Masouros, Christos and Petropulu, Athina P. and Hanzo, Lajos},
  journal={IEEE Transactions on Wireless Communications}, 
  title={Cooperative {ISAC} Networks: Performance Analysis, Scaling Laws, and Optimization}, 
  year={2025},
month={Feb.},
  volume={24},
  number={2},
  pages={877-892}
}

@ARTICLE{10091198,
	author={Wang, Xiao-Yang and Yang, Shaoshi and Yuan, Tian-Hao and Zhai, Hou-Yu and Zhang, Jianhua and Hanzo, Lajos},
	journal={IEEE Transactions on Vehicular Technology}, 
	title={High-Performance Low-Complexity Hierarchical Frequency Synchronization for Distributed Massive {MIMO-OFDMA} Systems}, 
	year={2023},
	month={Sep.},
	volume={72},
	number={9},
	pages={12343-12348},
	publisher={IEEE}
}

@techreport{3gpp.38.214, author = {{3GPP}}, institution = {{3rd Generation Partnership Project}}, month = {Dec.},  number = {38.214 {V}18.1.0}, title = {{Physical layer procedures for data}}, type={TS}, year = {2023}
}

@INPROCEEDINGS{9062788,
	author={Ozkaptan, Ceyhun D. and Ekici, Eylem and Altintas, Onur},
	booktitle={2019 IEEE Vehicular Networking Conference (VNC)}, 
	title={Demo: A Software-Defined {OFDM} Radar for Joint Automotive Radar and Communication Systems}, 
	year={2019},
	month={Dec.},
	pages={1-2},
	address={Los Angeles, USA}
}

@INPROCEEDINGS{8628347,
	author={Ozkaptan, Ceyhun D. and Ekici, Eylem and Altintas, Onur and Wang, Chang-Heng},
	booktitle={2018 IEEE Vehicular Networking Conference (VNC)}, 
	title={{OFDM} Pilot-Based Radar for Joint Vehicular Communication and Radar Systems}, 
	year={2018},
	month={Dec.},
	pages={1-8},
	address={Taipei, Taiwan}
}

@ARTICLE{10288116,
	author={Hu, Yanmo and Deng, Weibo and Zhang, J. Andrew and Guo, Y. Jay},
	journal={IEEE Wireless Communications Letters}, 
	title={Resource Optimization for Delay Estimation in Perceptive Mobile Networks}, 
	year={2024},
	month={Jan.},
	volume={13},
	number={1},
	pages={223-227},
	publisher={IEEE}
}

@ARTICLE{10103813,
	author={Huang, Yixuan and Yang, Jie and Tang, Wankai and Wen, Chao-Kai and Xia, Shuqiang and Jin, Shi},
	journal={IEEE Transactions on Wireless Communications}, 
	title={Joint Localization and Environment Sensing by Harnessing {NLOS} Components in {RIS}-Aided mmWave Communication Systems}, 
	year={2023},
	month={Dec.},
	volume={22},
	number={12},
	pages={8797-8813},
	publisher={IEEE}
}

@ARTICLE{10557620,
  author={Tagliaferri, Dario and Manzoni, Marco and Mizmizi, Marouan and Tebaldini, Stefano and Virgilio Monti-Guarnieri, Andrea and Maria Prati, Claudio and Spagnolini, Umberto},
  journal={IEEE Journal on Selected Areas in Communications}, 
  title={Cooperative Coherent Multistatic Imaging and Phase Synchronization in Networked Sensing}, 
  year={2024},
month={Oct.},
  volume={42},
  number={10},
  pages={2905-2921}

}

@ARTICLE{9534682,
	author={Tong, Xin and Zhang, Zhaoyang and Wang, Jue and Huang, Chongwen and Debbah, Mérouane},
	journal={IEEE Journal of Selected Topics in Signal Processing}, 
	title={Joint Multi-User Communication and Sensing Exploiting Both Signal and Environment Sparsity}, 
	year={2021},
	month={Nov.},
	volume={15},
	number={6},
	pages={1409-1422},
	publisher={IEEE}
}

@ARTICLE{10640151,
  author={Wang, Xiao-Yang and Yang, Shaoshi and Zhai, Hou-Yu and Masouros, Christos and Andrew Zhang, J.},
  journal={IEEE Transactions on Communications}, 
  title={Windowing Optimization for Fingerprint-Spectrum-Based Passive Sensing in Perceptive Mobile Networks}, 
  year={2025},
  month={Feb.},
  volume={73},
  number={2},
  pages={1367-1382},
  publisher={IEEE}
  }

@ARTICLE{11075613,
  author={Zhai, Hou-Yu and Yang, Shaoshi and Yuan, Tian-Hao and Wang, Xiao-Yang and Tan, Jing-Sheng and Chen, Sheng},
  journal={IEEE Transactions on Vehicular Technology}, 
  title={Hybrid {CSI}-Based Hierarchical Beamforming for Flexible Duplex {MIMO} Systems}, 
  year={2025, early access},
month={Jul.},
  volume={},
  number={},
  pages={1-6}
}

@ARTICLE{wxy,
	author={Wang, Xiao-Yang and Yang, Shaoshi and Zhang, Jianhua and Masouros, Christos and Zhang, Ping},
	journal={IEEE Journal on Selected Areas in Communications}, 
	title={Clutter Suppression, Parameter Association,
	and Time-Frequency Synchronization for
	Ranging and Velocity Estimation in Perceptive
	Vehicular Networks}, 
	year={2024, accepted},
	month={Apr.},
	publisher={IEEE}
	
}

@book{scheffe1999analysis,
	title={The analysis of variance},
	author={Scheffe Henry},
	volume={72},
	year={1999},
	publisher={John Wiley \& Sons}
}

@ARTICLE{5466526,
  author={Godrich, Hana and Haimovich, Alexander M. and Blum, Rick S.},
  journal={IEEE Transactions on Information Theory}, 
  title={Target Localization Accuracy Gain in {MIMO} Radar-Based Systems}, 
  year={2010},
month={Jun.},
  volume={56},
  number={6},
  pages={2783-2803}
}

@book{bertsekas2014constrained,
  title={Constrained optimization and Lagrange multiplier methods},
  author={Bertsekas, Dimitri P},
  year={2014},
  publisher={Academic press}
}

@ARTICLE{9364752,
  author={Sadeghi, Mohammad and Behnia, Fereidoon and Amiri, Rouhollah and Farina, Alfonso},
  journal={IEEE Transactions on Signal Processing}, 
  title={Target Localization Geometry Gain in Distributed {MIMO} Radar}, 
  year={2021},
month={Feb.},
  volume={69},
  number={},
  pages={1642-1652},
}

@ARTICLE{meng2024network,
  author={Meng, Kaitao and Han, Kawon and Masouros, Christos and Hanzo, Lajos},
  journal={IEEE Transactions on Wireless Communications}, 
  title={Network-level {ISAC}: An Analytical Study of Antenna Topologies Ranging from Massive to Cell-Free {MIMO}}, 
  year={2025, early access},
  volume={},
  number={},
}

@ARTICLE{10663785,
  author={Xu, Dongfang and Xu, Yiming and Zhang, Xin and Yu, Xianghao and Song, Shenghui and Schober, Robert},
  journal={IEEE Communications Magazine}, 
  title={Interference Mitigation for Network-Level {ISAC}: An Optimization Perspective}, 
  year={2024},
month={Sep.},
  volume={62},
  number={9},
  pages={28-34},

}

@ARTICLE{10614082,
  author={Wei, Zhiqing and Liu, Haotian and Feng, Zhiyong and Wu, Huici and Liu, Fan and Zhang, Qixun and Du, Yucong},
  journal={IEEE Internet of Things Magazine}, 
  title={Deep Cooperation in {ISAC} System: Resource, Node and Infrastructure Perspectives}, 
  year={2024},
month={Nov.},
  volume={7},
  number={6},
  pages={118-125}
}

@ARTICLE{10735119,
  author={Meng, Kaitao and Masouros, Christos and Chen, Guangji and Liu, Fan},
  journal={IEEE Transactions on Wireless Communications}, 
  title={Network-Level Integrated Sensing and Communication: Interference Management and {BS} Coordination Using Stochastic Geometry}, 
  year={2024},
month={Dec.},
  volume={23},
  number={12},
  pages={19365-19381}
}

@book{boyd2004convex,
  title={Convex optimization},
  author={Boyd, Stephen and Vandenberghe, Lieven},
  year={2004},
  publisher={Cambridge university press}
}

@ARTICLE{10438975,
  author={Cui, Yue and Ding, Haichuan and Zhao, Lian and An, Jianping},
  journal={IEEE Wireless Communications}, 
  title={Integrated Sensing and Communication: A Network Level Perspective}, 
  year={2024},
month={Feb.},
  volume={31},
  number={1},
  pages={103-109}
}

@ARTICLE{10292797,
  author={Zhang, Jianhua and Wang, Jialin and Zhang, Yuxiang and Liu, Yameng and Chai, Zeyong and Liu, Guangyi and Jiang, Tao},
  journal={IEEE Communications Magazine}, 
  title={Integrated Sensing and Communication Channel: Measurements, Characteristics, and Modeling}, 
  year={2024},
month={Jun.},
  volume={62},
  number={6},
  pages={98-104}

}

@article{wei2022toward,
  title={Toward multi-functional {6G} wireless networks: Integrating sensing, communication, and security},
  author={Wei, Zhongxiang and Liu, Fan and Masouros, Christos and Su, Nanchi and Petropulu, Athina P},
  journal={IEEE Communications Magazine},
  volume={60},
  number={4},
  pages={65-71},
  year={2022},
month={Apr.},
  publisher={IEEE}
}

@ARTICLE{236507,
  author={Li, J. and Compton, R.T.},
  journal={IEEE Transactions on Signal Processing}, 
  title={Maximum likelihood angle estimation for signals with known waveforms}, 
  year={1993},
month={Sep.},
  volume={41},
  number={9},
  pages={2850-2862}
}

@book{revesz2014laws,
  title={The laws of large numbers},
  author={R{\'e}v{\'e}sz, P{\'a}l},
  volume={4},
  year={2014},
  publisher={Academic Press}
}

@misc{wang2024fundamentaltradeoffstimefrequencyresource,
      title={On the Fundamental Trade-Offs of Time-Frequency Resource Distribution in {OFDMA} {ISAC}}, 
      author={Xiao-Yang Wang and Shaoshi Yang and Kaitao Meng and Hou-Yu Zhai and Christos Masouros},
      year={2024},
      eprint={2407.12628},
      archivePrefix={arXiv},
      primaryClass={eess.SP},
      url={https://arxiv.org/abs/2407.12628}, 
}

@article{rahman2019framework,
	title={Framework for a perceptive mobile network using joint communication and radar sensing},
	author={Rahman, Md Lushanur and Zhang, J Andrew and Huang, Xiaojing and Guo, Y Jay and Heath, Robert W},
	journal={IEEE Transactions on Aerospace and Electronic Systems},
	volume={56},
	number={3},
	pages={1926-1941},
	year={2019},
	month={Sep.},
	publisher={IEEE}
}

@article{zhang2021enabling,
	title={Enabling joint communication and radar sensing in mobile networks-a survey},
	author={Zhang, J Andrew and Rahman, Md Lushanur and Wu, Kai and Huang, Xiaojing and Guo, Y Jay and Chen, Shanzhi and Yuan, Jinhong},
	journal={IEEE Communications Surveys \& Tutorials},
	volume={24},
	number={1},
	pages={306-345},
	year={2021},
	month={First quarter.},
	publisher={IEEE}
}

@article{liu2020super,
	title={Super-resolution range and velocity estimations with {OFDM} integrated radar and communications waveform},
	author={Liu, Yongjun and Liao, Guisheng and Chen, Yufeng and Xu, Jingwei and Yin, Yingzeng},
	journal={IEEE Transactions on Vehicular Technology},
	volume={69},
	number={10},
	pages={11659-11672},
	year={2020},
	month={Oct.},
	publisher={IEEE}
}

@article{zhang2021overview,
	title={An overview of signal processing techniques for joint communication and radar sensing},
	author={Zhang, J Andrew and Liu, Fan and Masouros, Christos and Heath, Robert W and Feng, Zhiyong and Zheng, Le and Petropulu, Athina},
	journal={IEEE Journal of Selected Topics in Signal Processing},
	year={2021},
	month={Nov.},
	volume={15},
	number={6},
	pages={1295-1315},
	publisher={IEEE}
}

@ARTICLE{8999605,
	author={Liu, Fan and Masouros, Christos and Petropulu, Athina P. and Griffiths, Hugh and Hanzo, Lajos},
	journal={IEEE Transactions on Communications}, 
	title={Joint Radar and Communication Design: Applications, State-of-the-Art, and the Road Ahead}, 
	year={2020},
	volume={68},
	number={6},
	pages={3834-3862},
	month={Jun.},
	publisher={IEEE}
}

\end{document}